\documentclass[journal]{IEEEtran}

\usepackage[utf8]{inputenc} 
\usepackage[T1]{fontenc}
\usepackage{url}
\usepackage{ifthen}
\usepackage{cite}
\usepackage{enumerate}
\usepackage[cmex10]{amsmath} 
\usepackage{cite}
\usepackage{graphicx}
\usepackage{textcomp}
\usepackage{dsfont}
\usepackage{mathtools}

\DeclarePairedDelimiter\floor{\lfloor}{\rfloor}
\usepackage{lipsum}
\usepackage{amssymb, amsmath, amsthm, amsfonts}
\usepackage{tikz}
\usepackage{xcolor}
\usetikzlibrary{trees}
\newtheorem{thm}{Theorem}[section]
\newtheorem{lem}[thm]{Lemma}
\newtheorem{prop}[thm]{Proposition}
\newtheorem{cor}[thm]{Corollary}
\newtheorem{conj}{Conjecture}[section]
\ifCLASSOPTIONcompsoc
\usepackage[caption=false, font=normalsize, labelfont=sf, textfont=sf]{subfig}
\else
\usepackage[caption=false, font=footnotesize]{subfig}
\fi

\def\BibTeX{{\rm B\kern-.05em{\sc i\kern-.025em b}\kern-.08em
		T\kern-.1667em\lower.7ex\hbox{E}\kern-.125emX}}

\def\v{\mathbf{v}}
\def\0{\mathbf{0}}

\def\E{\mathbb{E}}
\def\P{\mathbb{P}}
\def\R{\mathbb{R}}
\def\Z{\mathbb{Z}}

\def\cQ{\mathcal{Q}}
\def\cR{\mathcal{R}}

\def\G{\Gamma}

\begin{document}
	\title{Probabilistic Forwarding of Coded Packets on Networks} 
	
	\author{B.~R.~Vinay Kumar,~\IEEEmembership{Student Member,~IEEE},
	\ and \
		Navin Kashyap,~\IEEEmembership{Senior Member,~IEEE}
		\thanks{
		This work was presented in part at the 2019 IEEE International Symposium on Information Theory (ISIT 2019) held in Paris, France.} 
		\thanks{The authors are with the Department of Electrical Communication Engineering, 
			Indian Institute of Science, Bengaluru, India.
			Email: \{vinaykb, nkashyap\}@iisc.ac.in}}
	\maketitle

	\begin{abstract}
We consider a scenario of broadcasting information over a network of nodes connected by noiseless communication links. A source node in the network has some data packets to broadcast. It encodes these data packets into $n$ coded packets in such a way that any node in the network that receives any $k$ out of the $n$ coded packets will be able to retrieve all the original data packets. The source transmits the $n$ coded packets to its one-hop neighbours. Every other node in the network follows a probabilistic forwarding protocol, in which it forwards a previously unreceived packet to all its neighbours with a certain probability $p$. 
We say that the information from the source undergoes a ``near-broadcast'' if the expected fraction of nodes that receive at least $k$ of the $n$ coded packets is close to $1$. The forwarding probability $p$ is chosen so as to minimize the expected total number of transmissions needed for a near-broadcast. We study how, for a given $k$, this minimum forwarding probability and the associated expected total number of packet transmissions varies with $n$. We specifically analyze the probabilistic forwarding of coded packets on two network topologies: binary trees and square grids. For trees, our analysis shows that for fixed $k$, the expected total number of transmissions increases with $n$. On the other hand, on grids, a judicious choice of $n$ significantly reduces the expected total number of transmissions needed for a near-broadcast. Behaviour similar to that of the grid is also observed in other well-connected network topologies such as random geometric graphs and random regular graphs.
	\end{abstract} 
	
	\section{Motivation and related work}\label{intro}
	An ad-hoc network is a network of nodes which communicate with each other without relying on any centralized infrastructure.  A classical example of ad-hoc networks is wireless sensor networks (WSNs) which have sensors measuring temperature, humidity etc. connected with each other. The Internet of Things (IoT) network, which involves different types of physical devices --- sensors, actuators, routers, mobiles etc. ---  communicating with each other over a network can be thought of as an ad-hoc network.
	\par Broadcast mechanisms on such distributed networks are crucial in order to disburse key network-related information throughout the network. In the applications mentioned above, updation of sensing parameters in WSNs or over-the-air programming of the IoT nodes are done typically through a broadcast mechanism. These broadcasts are usually initiated from a single node in the network which is easily accessible (a mobile phone, say). In this paper, we will assume that there is a source node, $s$, which has $k_s$ packets of information which need to be broadcast in the network. A natural broadcast algorithm is flooding, wherein a node forwards every newly received packet to all its one-hop neighbours. If there are $N$ nodes in the network, then the total number of transmissions is $k_sN$. However, a node might receive the same packet from multiple neighbours resulting in wasteful transmissions. Moreover, flooding is also known to result in the `broadcast-storm' problem \cite{tseng2002broadcast}. In short, although the flooding mechanism is simple and easy to implement, there is an excessive number of transmissions in the network, resulting in a high energy expenditure. 
	\par For the applications that we are interested in, such a broadcast algorithm is not feasible since individual nodes are energy-constrained. Additionally, each node in the network has minimal computational ability and limited knowledge of the network topology. To adhere to these limitations, any broadcast algorithm that is proposed needs to be completely distributed, must minimize energy consumption, should run in finite time, and must impose minimal computational burden on the individual nodes. 
	\par Probabilistic forwarding as a broadcast mechanism, has been proposed in the literature (see  \cite{sasson2003probabilistic}) as an alternative to flooding. Here, each node, on receiving a packet for the first time, either forwards it to all its one-hop neighbours with probability $p$ or takes no action with probability $1-p$. Probabilistic forwarding has also been referred to as a gossip algorithm in \cite{haas2006gossip}, in which the authors claim a $35\%$ reduction in the transmission overhead as compared to flooding. An upper bound on the expected number of transmissions for this algorithm can be obtained thus. On a network of $N$ nodes, an average of $Np$ nodes decide to transmit a given source packet, irrespective of whether they receive it or not. Since there are $k_s$ source packets in all, there are $k_sNp$ expected total number of transmissions. For a forwarding probability $p<1$, this is less than the number of transmissions for flooding. Nevertheless, a drawback of probabilistic forwarding is that a particular node in the network may not receive one of the $k_s$ packets, and hence, is unable to obtain the information from the source. 
	\par In order to overcome this, we introduce coded packets along with probabilistic forwarding. We describe the setup here. Consider a large network with a particular node designated as the source $s$. The source has $k_s$ message packets to send to a large fraction of nodes in the network. The $k_s$ message packets are first encoded into $n$ coded packets such that, for some $k \ge k_s$, the reception of any $k$ out of the $n$ coded packets by a node suffices to retrieve the original $k_s$ message packets. Examples of codes with this property are Maximum Distance Separable (MDS) codes ($k=k_s$), fountain codes ($k=k_s(1+\epsilon)$ for some $\epsilon >0$) etc. which are used in practice. The $n$ coded packets are indexed using integers from $1$ to $n$, and the source transmits each packet to all its one-hop neighbours. All the other nodes in the network use the probabilistic forwarding mechanism: when a packet (say, packet $ \# j$) is received by a node for the first time, it either transmits it to all its one-hop neighbours with probability $p$ or does nothing with probability $1-p.$ The node ignores all subsequent receptions of packet $\# j$. Packet collisions and interference effects are neglected.
	
	Our goal is to analyze the performance of the above algorithm. In particular, we wish to find the minimum retransmission probability $p$ for which the expected fraction of nodes receiving at least $k$ out of the $n$ coded packets is close to 1, which we deem a ``near-broadcast''. 
This probability yields the minimum value for the expected total number of transmissions across all the network nodes needed for a near-broadcast. The expected total number of transmissions is taken to be a measure of the energy expenditure in the network. 

\begin{figure}
	\includegraphics[width=\linewidth]{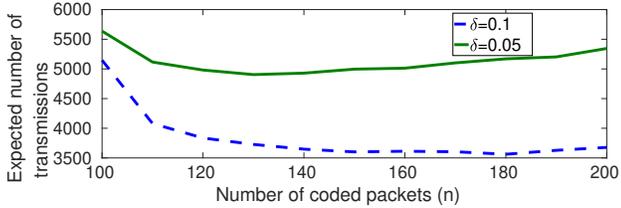}
	\caption{Expected total number of transmissions on a RGG with $60$ nodes in a $20\times 20$ square with two nodes being connected if they are at most $r=5.5$ apart. The forwarding probability $p$ is such that the expected fraction of nodes that receive at least $k=100$ of the $n$ coded packets is at least $1-\delta$.}
	\label{fig:rgg}
	\vspace{-1.5em}
\end{figure}
	
	Simulation results presented in \cite{ncc2018:probfwding} indicate that over a wide range of network topologies (including the important case of random geometric graphs (RGGs), but not including tree-like topologies),  the expected total number of transmissions initially decreases to a minimum and then gradually increases with $n$. A representative simulation result on a RGG is provided in Fig. \ref{fig:rgg}. Our aim is to understand this behaviour and predict, via analysis, the value of $n$ that minimizes the expected number of transmissions. We would ultimately like to explain this behaviour on random geometric graphs, which constitute an important model for wireless ad-hoc networks \cite{vaze2015random}. However, we have not yet developed the tools required for the analysis there. 
	\par Our work is closest in spirit to the works in \cite{haas2006gossip} and \cite{Shen2006dirbroad}. In \cite{haas2006gossip}, the authors describe variants of the GOSSIP protocol and provide heuristics and simulation results for improving flooding and routing mechanisms in networks. The authors in \cite{Shen2006dirbroad} map the broadcast mechanism to percolation on networks which is the approach that we will be using in this paper when the underlying network has a grid topology. The authors in \cite{Shen2006dirbroad} further use directional antennas to reduce the transmission overhead. To the best of our knowledge, none of the previous works have looked at gossip algorithms when there are coded packets. Our contributions in this paper comprise of a detailed analysis of the probabilistic forwarding mechanism with coded packets when the underlying network topologies are binary trees and grids. The analysis on the tree is straightforward and uses concentration bounds on binomial random variables. The case of the grid is far more interesting with the arguments involving ideas from ergodic theory and the site percolation process on $\mathbb{Z}^2$. The mapping between connectivity in finite networks and the site percolation process has been discussed briefly in \cite[Chapter 3]{franceschetti2008random}. Our approach builds on this mapping to obtain estimates of the minimum forwarding probability and the expected total number of transmissions for the grid.
	\par The rest of the paper is organized as follows. Section \ref{sec:formulation} contains a mathematical formulation of the problem. In Section \ref{sec:init_obs}, some initial observations are made on the minimum forwarding probability for a near-broadcast and the expected total number of transmissions at this probability. Some simulation results are provided to support these observations as well. In Section \ref{sec:tree} we consider the problem on rooted binary trees and derive expressions for the minimum forwarding probability and the expected total number of transmissions. We show that probabilistic forwarding using coded packets is not beneficial on trees. In Section \ref{sec:grid}, we provide estimates for the expected number of transmissions on the grid. Ergodic theory and the theory of site percolation are used to obtain these estimates. Section \ref{sec:discussion} discusses some critical aspects of the analysis and the behaviour of the probabilistic forwarding on other graph topologies. It also briefly describes certain extensions of the algorithm that are possible. The appendix contains some auxiliary results needed for our analysis.

\section{Problem Setting}\label{sec:formulation}
Consider a graph $G=(V,E)$, where $V$ is the vertex set with $N$ vertices (nodes) and $E$ is the set of edges (noiseless communication links). It is assumed that when a node broadcasts a packet, all its one-hop neighbours receive the packet without any errors. A source node $s \in V$ has a certain number of message packets which need to be broadcast in the network. The source $s$ encodes these messages into $n$ coded packets in such a way that a node that receives any $k$ of the $n$ coded packets can retrieve all the original message packets. It is assumed that each packet has a header which identifies the packet index $j \in [n] := \{1,2,\ldots,n\}$.  The source node broadcasts all $n$ coded packets to its one-hop neighbours, after which the probabilistic forwarding protocol takes over. A node receiving a particular packet for the first time, forwards it to all its one-hop neighbours with probability $p$ and takes no action with probability $1-p$. Each packet is forwarded independently of other packets and other nodes. This probabilistic forwarding continues until there are no further transmissions in the system. The protocol indeed must terminate after finitely many transmissions since each node in the network may choose to forward a particular coded packet only the first time it is received. The node ignores all subsequent receptions of the same packet, irrespective of the decision it took at the time of first reception.

We are interested in the following scenario. Let $\mathcal{R}_{k,n}$ be the nodes, including the source node, that receive at least $k$ out of the $n$ coded packets. We call these \textit{successful receivers} and denote the number of such nodes by $R_{k,n}$. Given a $\delta \in (0,1)$, let\footnote{The quantities $R_{k,n}$, $p_{k,n,\delta}$, $\tau_{k,n,\delta}$ etc.\ are all, of course, functions of the underlying graph $G$ as well, but for simplicity, we usually suppress this dependence from our notation. We use $R_{k,n}(G)$, $p_{k,n,\delta}(G)$, $\tau_{k,n,\delta}(G)$ etc.\ whenever the dependence on $G$ needs to be made explicit.} $p_{k,n,\delta}$ be the minimum forwarding probability $p$ for a near-broadcast, i.e.,
	\begin{equation}
	p_{k,n,\delta} \ := \ \inf\left\{p\  \bigg|\ \mathbb{E}\left[\frac{R_{k,n}}{N}\right] \geq 1-\delta \right\}.
	\label{def:pkndelta}
	\end{equation} 

The performance measure of interest, denoted by $\tau_{k,n,\delta}$, is the expected total number of transmissions across all nodes when the forwarding probability is set to $p_{k,n,\delta}$. Here, it should be clarified that whenever a node forwards (broadcasts) a packet to all its one-hop neighbours, it is counted as a single (simulcast) transmission. Our aim is to determine, for a given $k$ and $\delta$, how $\tau_{k,n,\delta}$ varies with $n$, and the value of $n$ at which it is minimized (if it is indeed minimized). To this end, it is necessary to first understand the behaviour of $p_{k,n,\delta}$ as a function of $n$. In the next section, we make some initial observations for the minimum retransmission probability, $p_{k,n,\delta}$, and the corresponding value of the expected total number of transmissions, $\tau_{k,n,\delta}$, valid for any underlying network topology.

\section{Initial observations}\label{sec:init_obs}
On any connected graph $G=(V,E)$, when a successful receiver must receive $k$ out of $n'$ coded packets, instead of $k$ out of $n$, where $n'>n$, each packet can be transmitted at a lower probability while still ensuring a near-broadcast. In fact, the minimum forwarding probability goes to $0$ as $n$ is increased. This is formalized in the following lemma.
\begin{lem}\label{lem:pkndelzero}
	For fixed values of $k$ and $\delta$, 
	\begin{enumerate}
		\item[\emph{(a)}] $p_{k,n,\delta}$ is a non-increasing function of n.
		\item[\emph{(b)}] $p_{k,n,\delta} \rightarrow 0$ as $n\rightarrow \infty.$
	\end{enumerate}
\end{lem}
\begin{IEEEproof}\label{sec:lem1proof}
(a)\ For any $n > 0$, the random variables $R_{k,n}$ and $R_{k,n-1}$ can be coupled as follows: If there are a total of $n$ coded packets, then $R_{k,n-1}$ (resp.\ $R_{k,n}$) is realized as the number of nodes, including the source node, that receive at least $k$ of the \emph{first} $n-1$ (resp.\ at least $k$ of the $n$) coded packets. It is then clear that $\E[\frac{1}{N} R_{k,n}] \ge \E[\frac{1}{N} R_{k,n-1}]$, and hence, by \eqref{def:pkndelta}, we have $p_{k,n,\delta} \le p_{k,n-1,\delta}$.

(b)\  From the $n$ coded packets, create $\lfloor\frac{n}{k}\rfloor$ non-overlapping (i.e., disjoint)  groups of $k$ packets each. For $i=1,2,\cdots,\floor{\frac{n}{k}}$, let $A_i$ be the event that the $i$th group of $k$ coded packets is received by at least $(1-\delta/2)N$ nodes. The events $A_i$ are mutually independent and have the same probability of occurrence. For any $p>0$, we have $\P(A_i)$ being strictly positive (but perhaps small). Hence, $$\P(\text{at least one } A_i \text{ occurs}) = 1-\bigl(1-\P(A_1)\bigr)^{\lfloor\frac{n}{k}\rfloor} \ge 1-\frac{\delta}{2}$$ for all sufficiently large $n$, so that $\P\left(\frac{R_{k,n}}{N}\ge 1-\delta/2\right)\ge1-\delta/2.$ This further implies that $\frac{\E\left[R_{k,n}\right]}{N}\ge (1-\delta/2)(1-\delta/2) \ge 1-\delta$. Thus, for any $p>0$, we have $p_{k,n,\delta}\le p$ for all sufficiently large $n$.
\end{IEEEproof}

\begin{figure} 
	\centering
	\subfloat[Minimum retransmission probability]{%
		\includegraphics[width=\linewidth]{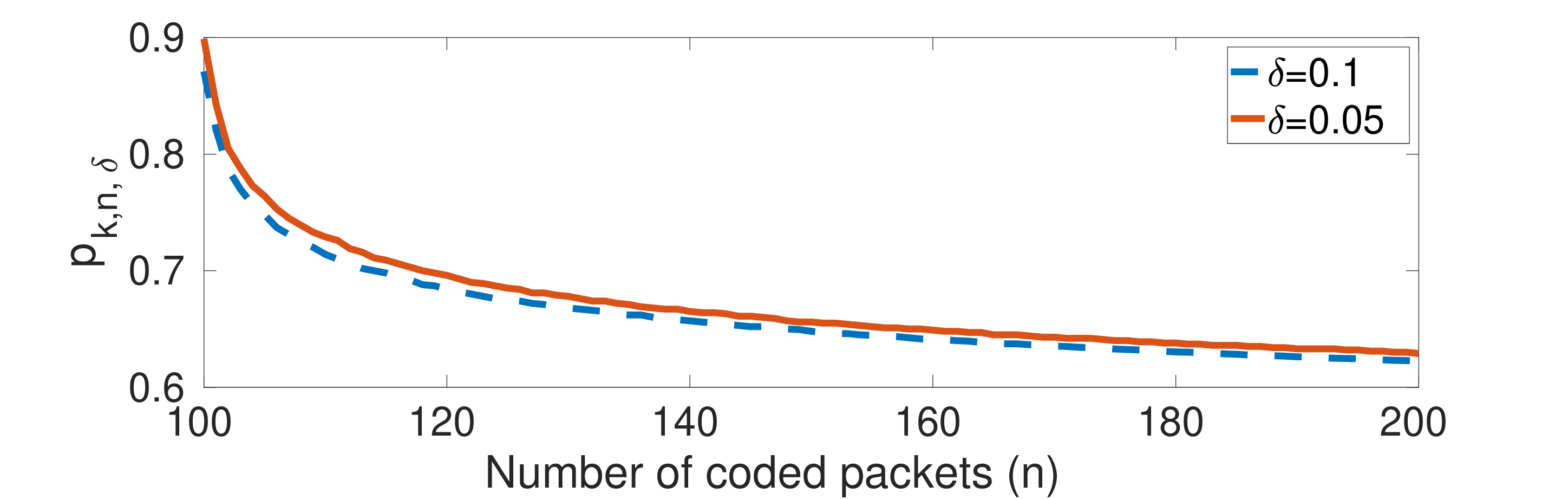}}
	\label{4a}
	\subfloat[Expected total number of transmissions]{%
		\includegraphics[width=\linewidth]{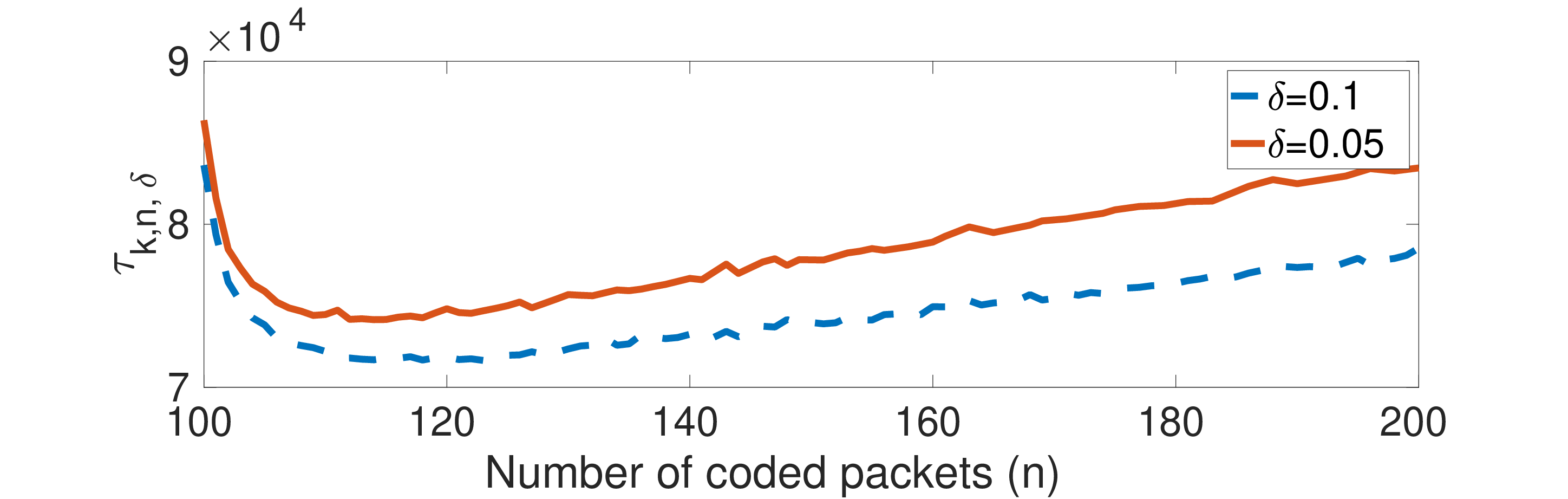}}
	\label{4b}\\
	\caption{Simulation on a $31 \times 31$ grid with k=100 packets}
	\label{fig:grid_simu} 
\end{figure}
Fig. \ref{fig:grid_simu} gives simulation results for the minimum forwarding probability and the expected total number of transmissions on a $31 \times 31$ grid with $k=100$ packets. Notice that the minimum forwarding probability, $p_{k,n,\delta}$, decreases with $n$ as proved above. On the other hand, the expected total number of transmissions, $\tau_{k,n,\delta}$, typically exhibits more complex behaviour. Over a wide range of graph topologies (both deterministic and random), except notably for trees (see Section~\ref{sec:tree}), $\tau_{k,n,\delta}$ initially decreases and then grows gradually as $n$ increases. This trend was seen on a RGG in Fig. \ref{fig:rgg} and is more pronounced for a grid topology --- see Fig. \ref{fig:grid_simu}(b). Thus, there typically is an optimal value of $n$ that minimizes $\tau_{k,n,\delta}$. This means that the ad-hoc network needs to be operated at this value of the number of coded packets $n$ and the corresponding forwarding probability $p_{k,n,\delta}$, in order to have least energy expenditure overall. Notice also that probabilistic forwarding with no coding corresponds to the point $n=k=100$ packets in Fig. \ref{fig:grid_simu}. The number of transmissions $\tau_{k,n,\delta}$ decreases (initially) when coded packets are introduced which highlights the advantage of coding with probabilistic forwarding on such network topologies.
\par The decrease in $\tau_{k,n,\delta}$ happens due to an interplay between two opposing factors: as $n$ increases, $p_{k,n,\delta}$ decreases (Lemma~\ref{lem:pkndelzero}), which contributes towards a decrease in $\tau_{k,n,\delta}$. But this is opposed by the fact that the overall number of transmissions tends to increase when there are more number of packets traversing the network. 
\par To determine the value of $n$ that minimizes $\tau_{k,n,\delta}$, we need more precise estimates of $p_{k,n,\delta}$, and consequently, $\tau_{k,n,\delta}$. For specific graph topologies, we may be able to obtain such estimates using methods tailored to those topologies. We demonstrate this for two topologies in the next two sections, starting with the easiest case of a binary tree.
	
\section{Rooted Binary Trees}\label{sec:tree}
\begin{figure}
	\centerline{\scalebox{0.38}{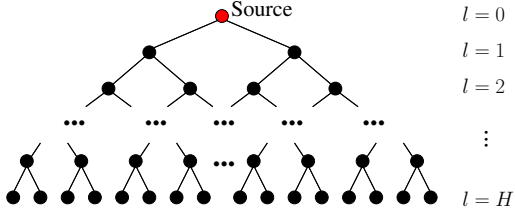}} 		
	\caption{A rooted binary tree of height $H$.}
	\label{fig:bintree}
	\vspace{-1.5em}
\end{figure}
\begin{figure} 
	\centering
	\subfloat[Minimum retransmission probability]{%
		\includegraphics[width=\linewidth]{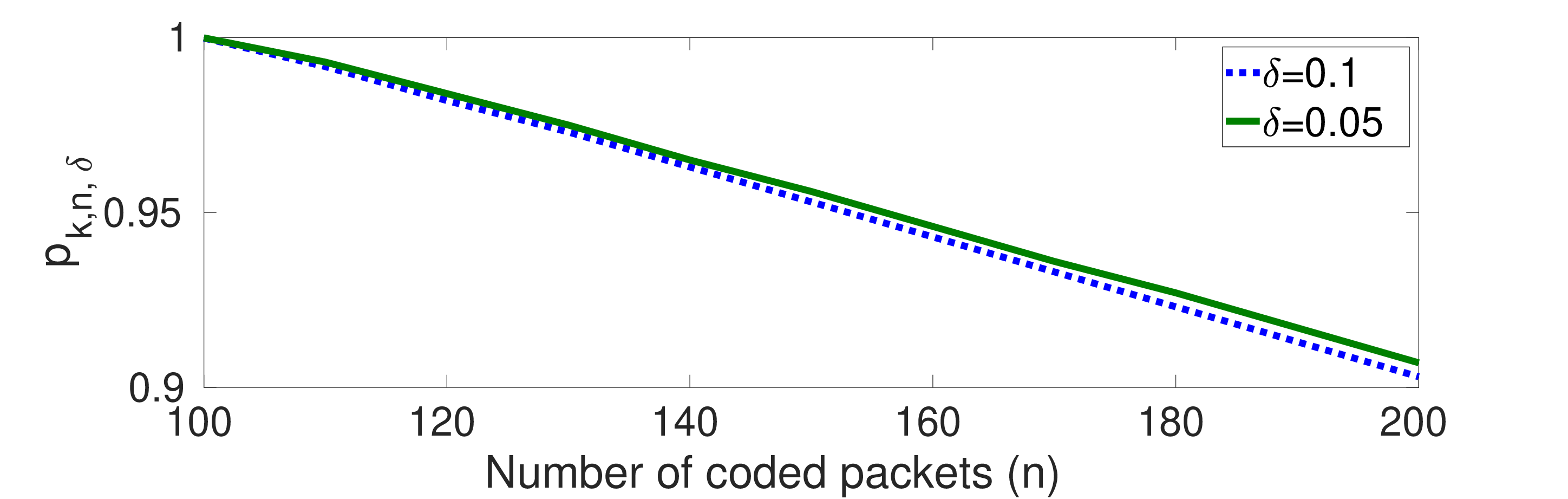}}
	\label{3a}
	\subfloat[Expected total number of transmissions]{%
		\includegraphics[width=\linewidth]{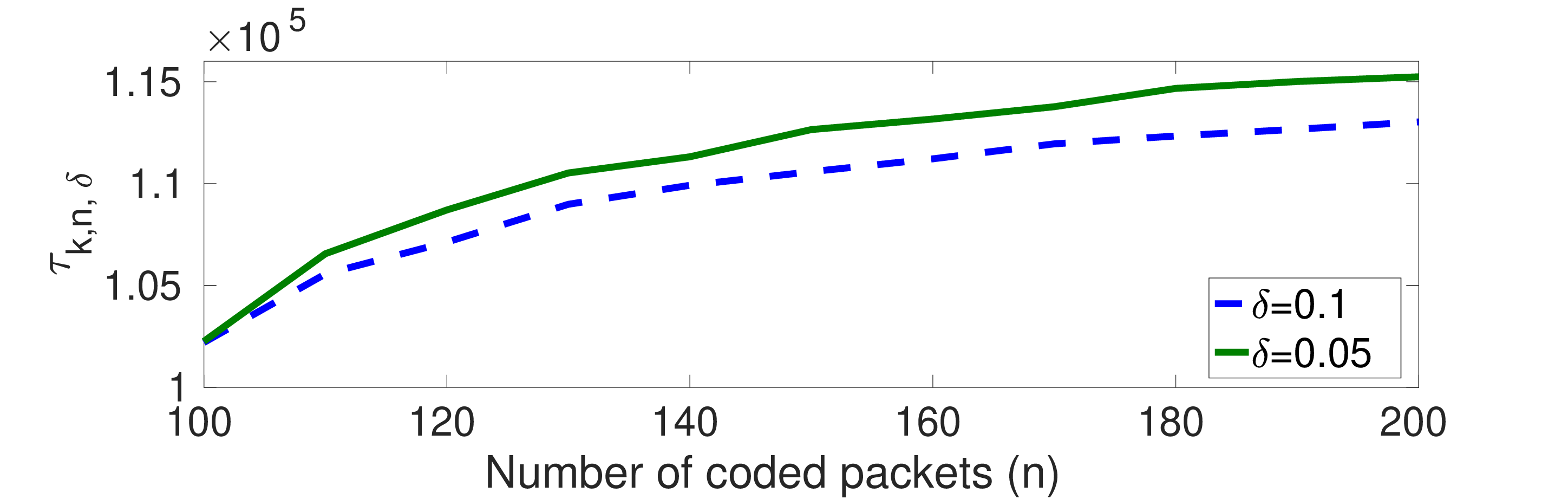}}
	\label{3b}\\
	\caption{Probabilistic forwarding on a binary tree of height $H=10$.}
	\label{fig:treesimu} 
	\vspace{-1.5em}
\end{figure}
Consider a rooted binary tree of height $H \ge 2$ as depicted in Fig.~\ref{fig:bintree}. Simulation results from running the probabilistic forwarding protocol on a binary tree of height $H=10$ with $k=100$ packets and $n$ between $100$ and $200$ are shown in Fig. \ref{fig:treesimu}. The minimum forwarding probability decreases as the number of coded packets $n$ is increased, as was proved in Lemma \ref{lem:pkndelzero}. The expected total number of transmissions however, increases monotonically when coded packets are introduced, unlike the trend that is discussed in the previous section. Thus, introducing coded packets along with probabilistic forwarding does not help in reducing the number of transmissions when the underlying network has a tree-like structure. In this section, we analyze the probabilistic forwarding mechanism on the binary tree and show that this is indeed true. The analysis in this section extends easily to the case of rooted $d$-ary trees, for any $d \ge 2$.

In a binary tree of height $H$, there are $2^\ell$ nodes at level $\ell$, for $\ell=0, 1, 2,\cdots, H$, and hence the total number of nodes in the network are $N=2^{H+1}-1$. The root of the tree at level $\ell=0$ is the source node and it encodes its data packets into $n$ coded packets and transmits them to its children. Every other node on the tree follows the probabilistic forwarding strategy with some fixed forwarding probability $p > 0$. Nodes that share a common parent receive	the same packets and hence will possess the same number of packets at the end of the probabilistic forwarding mechanism. We will assume that the nodes at level $H$ (i.e., the leaf nodes) do not transmit, as there is nothing to be gained in allowing them to do so. 

To get a handle on the minimum retransmission probability $p_{k,n,\delta}$ for a near-broadcast, we first look at the number of successful receivers, $R_{k,n}$. We can write $R_{k,n}=\sum_{\ell=0}^{H}R_\ell$, where $R_\ell$ is the number of
nodes at level $\ell$ that hold at least $k$ of the $n$ packets. Similarly, define $T_{k,n}=\sum_{\ell=0}^{H}T_\ell$, where $T_\ell$ is the number of transmissions by nodes at level $\ell$. Note that $T_0=n$ and $R_0=1$ since the source transmits all the $n$ packets. Also, $T_H=0$ from our assumption that leaf nodes do not transmit any packet.

In a tree, there is only a single path from the root to any
node in the tree. Thus, for a node $\v$ at level $\ell$ to receive a packet from the root, all the intermediate nodes on
the unique path from the root to $\v$ need to transmit the packet. Hence, for $\ell \ge 1$,
$$\P(\text{node }\v\text{ at level }\ell \text{ receives the }j\text{th packet})=p^{\ell-1}.$$
Since individual packets are transmitted independently of each other, we have
\begin{IEEEeqnarray}{rCl}\label{prob_nodes}
	\P(\text{node $\v$} &\ & \!\!\!\!\!\!\!\! \text{ at level $\ell$ receives at least $k$ out of $n$ packets}) \nonumber \\
	&=& \sum_{r=k}^{n} \binom{n}{r}p^{(\ell-1)r}(1-p^{\ell-1})^{n-r} \nonumber \\
	&=& \P(Z_{\ell-1} \ge k), \nonumber
\end{IEEEeqnarray}
where $Z_{\ell-1} \sim \mathrm{Bin}(n,p^{\ell-1})$ is a binomial random variable with parameters $n$ and $p^{\ell-1}$. Summing the above over all nodes $\v$ at level $\ell$, we obtain $\E[R_\ell] = 2^\ell \,\P(Z_{\ell-1} \ge k)$, and hence, 
\begin{equation}\label{rCl}
\E[R_{k,n}] \ = \  1+\mathbb{E}\left[\sum_{\ell=1}^{H}R_\ell\right] \ = \ 1+ \sum_{\ell=1}^{H}2^\ell \, \P(Z_{\ell-1}\ge k).
\end{equation}
Similarly,  a node $\v$ at level $\ell \in \{0,1,\cdots , H-1\}$ receives a packet from the source and transmits it with probability $p^\ell$. This gives the total expected number of transmissions for a transmission probability $p$ to be
$$
\E[T_{k,n}] \ = \ \sum_{\ell=0}^{H-1} \E[T_\ell] \ = \  n\frac{(2p)^H-1}{2p-1}.
$$
Thus, $\E[T_{k,n}]$ is a monotonically increasing function of $p$, from which it can be inferred that
\begin{equation}
\tau_{k,n,\delta}=n \frac{(2p_{k,n,\delta})^{H}-1}{2p_{k,n,\delta}-1}.
\label{eq:tau_tree}
\end{equation}
$p_{k,n,\delta}$ is the minimum probability such that $\mathbb{E}\left[\frac{R_{k,n}}{N}\right] \geq 1-\delta$. From our computation above and the expression for $N$, it can be written as 
$$	p_{k,n,\delta} \ := \ \inf\left\{p\  \bigg|\ \frac{1+ \sum_{\ell=1}^{H}2^\ell \, \P(Z_{\ell-1}\ge k)}{2^{H+1}-1} \geq 1-\delta \right\},$$
where $Z_\ell \sim \mathrm{Bin}(n,p^\ell)$ for $\ell = 0,1,\ldots,H-1$.
The inequality inside the parantheses can be rewritten as
	\begin{equation}\label{exprtomin}
	\frac{\sum_{\ell=0}^{H-1}2^{\ell+1}\mathbb{P}(Z_\ell \leq k-1)}{2^{H+1}-1} \leq \delta.
	\end{equation}
An analysis starting from \eqref{exprtomin} yields the two propositions below, which provide good lower and upper bounds on $p_{k,n,\delta}$. These bounds are plotted, for $k=100$, $\delta = 0.1$ and $H = 50$, in Fig.~\ref{fig:tree_plots}(a) along with the exact values of $p_{k,n,\delta}$ obtained numerically from  \eqref{exprtomin}. The corresponding plots for $\tau_{k,n,\delta}$, obtained via  \eqref{eq:tau_tree}, are shown in Fig.~\ref{fig:tree_plots}(b). 
	
	\begin{prop}
		Let $k \ge 2$, $H \ge 2$, and $0 \le \delta < \frac18$ be fixed. For all $n \ge k$, we have 
		$
		p_{k,n,\delta} > {\left(\frac{k-1}{n}\right)}^{\frac{1}{H-1}}.
		$
		
		In the case of $k=1$ and $n > 1$, the lower bound can be improved to
		$
		p_{k,n,\delta} > {\left(\frac{1}{n}\right)}^{\frac{1}{H-1}}.
		$
		\label{prop:p_lobnd}
	\end{prop}
		\begin{IEEEproof}
			Suppose that $p$ is such that $np^{H-1} \le k-1$. Then, $Z_{H-1}$ has mean at most $k-1$. As a result, the median of $Z_{H-1}$ is also at most $k-1$ \cite[Corollary~3.1]{jogdeo1968monotone}. In other words, $\P(Z_{H-1} \le k-1) \ge \frac12$. Consequently, $\sum_{l=0}^{H-1}2^{l+1}\mathbb{P}(Z_l \leq k-1) \ge 2^H \mathbb{P}(Z_{H-1} \le k-1) \ge 2^{H-1}$, so that the left-hand side (LHS) of \eqref{exprtomin} is at least $\frac{2^{H-1}}{2^{H+1}-1} \ge \frac{2^{H-1}}{2^{H+1}} = 0.25 > \delta$. Hence, for \eqref{exprtomin} to hold, we must have $np^{H-1} > k-1$, from which the lower bound on $p_{k,n,\delta}$ follows.
			
			In the case of $k=1$, suppose that $p \le {\left(\frac{1}{n}\right)}^{H-1}$. Then, $\mathbb{P}(Z_{H-1} = 0) = (1-p^{H-1})^n \ge (1-\frac{1}{n})^n \ge (1-\frac{1}{2})^2 = 0.25$, for all $n \ge 2$. Hence, $\sum_{l=0}^{H-1}2^{l+1}\mathbb{P}(Z_l \leq k-1) \ge 2^H \mathbb{P}(Z_{H-1} = 0) \ge 2^{H-2}$. As a result, the LHS of \eqref{exprtomin} is at least $\frac{2^{H-2}}{2^{H+1}} = 0.125 > \delta$. Thus, again, for \eqref{exprtomin} to hold, we need $p > {\left(\frac{1}{n}\right)}^{H-1}$.	
		\end{IEEEproof}
		
	\medskip
	
	\begin{prop}
		Let $k \ge 2$, $H \ge 2$, and $0 < \delta \le 1$ be fixed, and let $\delta' :=  \min\left\{\delta  \left(\frac{2^{H+1}-1}{2^{H+1}-2}\right), 1 \right\}$. Then, for all $n \ge 1$, we have
		$$
		p_{k,n,\delta} \le \min\left\{{\left(\frac{k-1+t}{n}\right)}^{\frac{1}{H-1}},1\right\}, 
		$$
		where $t = \sqrt{2(k-1)(-\ln\delta') + (\ln\delta')^2} - \ln\delta'$.
		In the case of $k=1$, the bound
		$$
		p_{k,n,\delta} \le \min\left\{{\left(\frac{-\ln\delta'}{n}\right)}^{\frac{1}{H-1}}, 1\right\}
		$$
		holds for all $n \ge 1$.
		\label{prop:p_upbnd}
	\end{prop}
			
			\begin{IEEEproof}
				Note first that for all $l \le H-1$, we have\footnote{This is easily shown by a standard coupling argument --- see e.g., \cite[Lemma~IV.1]{ncc2018:probfwding}.} $\mathbb{P}(Z_l \le k-1) \le \mathbb{P}(Z_{H-1} \le k-1)$. Hence, $\sum_{l=0}^{H-1}2^{l+1}\mathbb{P}(Z_l \leq k-1) \le \bigl(\sum_{l=0}^{H-1} 2^{l+1}\bigr) \mathbb{P}(Z_{H-1} \le k-1) = (2^{H+1}-2) \mathbb{P}(Z_{H-1} \le k-1)$. Thus, to show that \eqref{exprtomin} holds, it suffices to prove that $\mathbb{P}(Z_{H-1} \le k-1) \le \delta \, \left(\frac{2^{H+1}-1}{2^{H+1}-2}\right)$. It is, therefore, enough to show that $\mathbb{P}(Z_{H-1} \le k-1) \le \delta'$.
				
				Consider $k = 1$ first. Take $p = \min\left\{1,{\left(\frac{C'}{n}\right)}^{\frac{1}{H-1}}\right\}$, where $C' = -\ln\delta'$. Then, $\mathbb{P}(Z_{H-1} \le k-1) = \mathbb{P}(Z_{H-1} = 0) = (1-p^{H-1})^n$, which, by choice of $p$, is either equal to $0$ (if $C' \ge n$) or $(1-C'/n)^n$ (if $C' < n$). In either case, $\mathbb{P}(Z_{H-1} = 0) $ is less than $e^{-C'} = \delta'$, as needed.
				
				Consider $k \ge 2$ now. Take $p = \min\left\{1,{\left(\frac{k-1+t}{n}\right)}^{\frac{1}{H-1}}\right\}$, where $t$ is as in the statement of the proposition. For $n \ge k-1+t$, we have $Z_{H-1} \sim \mathrm{Bin}(n,\frac{k-1+t}{n})$, so that \begin{align*}
				\mathbb{P}(Z_{H-1} \le k-1) & = \mathbb{P}\bigl(Z_{H-1} \le n({\textstyle \frac{k-1+t}{n} - \frac{t}{n}})\bigr) \notag \\
				& \le  e^{-n \, D(\frac{k-1}{n} \parallel \frac{k-1+t}{n})}
				\end{align*}
				via the Chernoff bound. Here, $D(\cdot \parallel \cdot)$ denotes the Kullback-Leibler divergence, defined as $D(x \parallel y) = x \ln \frac{x}{y} + (1-x) \ln \frac{1-x}{1-y}$. Using the bound $D(x \parallel y) \ge \frac{(x-y)^2}{2y}$, valid for $x \le y$ \cite{Okamoto1959}, we further have
				$$
				\mathbb{P}(Z_{H-1} \le k-1) \le e^{-n\left[\frac{(t/n)^2}{2(k-1+t)/n}\right]} = e^{-\frac{t^2}{2(k-1+t)}}.
				$$
				Thus, to conclude that $\mathbb{P}(Z_{H-1} \le k-1) \le \delta'$, as required, it suffices to show that $\frac{t^2}{2(k-1+t)} \ge -\ln \delta'$. This can be re-written as $t^2 + 2t \ln\delta' + 2(k-1)\ln\delta' \ge 0$, or equivalently, $(t+\ln\delta')^2 + 2(k-1)\ln\delta' - (\ln\delta')^2 \ge 0$, which is evidently satisfied by our choice of $t$. 	
			\end{IEEEproof}
			
 The following theorem, which summarizes the behaviour of $p_{k,n,\delta}$ on binary trees, is a direct consequence of Propositions~\ref{prop:p_lobnd} and~\ref{prop:p_upbnd}. 
 
	\begin{thm}\label{thm:pkndelta}
		Let $k \ge 2$, $H \ge 2$ and $0 < \delta < \frac18$ be fixed. We then have
		$
		p_{k,n,\delta} = \Theta\left({\textstyle{\left(\frac{k}{n}\right)}^{\frac{1}{H-1}}}\right),
		$
		where the constants implicit in the $\Theta$-notation\footnote{The notation $a(n) = \Theta(b(n))$ means that there are positive constants $c_1$ and $c_2$ such that $c_1 b(n) \le a(n) \le c_2 b(n)$ for all sufficiently large $n$.} may be chosen to depend only on $H$ and $\delta$.
	\end{thm}
	
Tighter bounds for $p_{k,n,\delta}$ can be obtained by bounding the binomial cumulative distributive function (CDF) in  \eqref{exprtomin} using Theorem \ref{bin_normal_cdf} of the appendix. This gives, 
\begin{equation}
p_{k,n,\delta} \le \inf \left\{p \ \Bigg| \frac{\sum_{\ell=0}^{H-1}2^{\ell+1} C_{n,p^\ell}(k)}{2^{H+1}-1} \le \delta \right\}
\label{Eq:normal_tree_upper}
\end{equation}
and
\begin{equation}
p_{k,n,\delta} \ge \inf \left\{p \ \Bigg| \frac{\sum_{\ell=0}^{H-1}2^{\ell+1} C_{n,p^\ell}(k-1)}{2^{H+1}-1}\le \delta\right\},
\label{Eq:normal_tree_lower}
\end{equation}
where $C_{n,q}(k) = \Phi \left(\mathrm{sgn}\left(\frac{k}{n}-q\right)\sqrt{2nD(\frac{k}{n} \left|\right|  q)}\right)$.
The plots in Fig.~\ref{fig:tree_plots} provide a theoretical explanation for why $\tau_{k,n,\delta}$ increases with $n$. Another confirmation of this behaviour can be obtained by substituting $p_{k,n,\delta} = c {\bigl(\frac{k}{n}\bigr)}^{\frac{1}{H-1}}$, for a suitable positive constant $c \equiv c(H,\delta)$, into the expression for $\tau_{k,n,\delta}$ in  \eqref{eq:tau_tree}. This yields 
	\begin{equation*}
	\tau_{k,n,\delta}  \ =  \ \frac{\left(n^{\frac{1}{H-1}}\right)^H-\kappa^H}{
	n^{\frac{1}{H-1}}-\kappa}, \label{taukrho2}
	\end{equation*}
	where $\kappa=2ck^{\frac{1}{H-1}}$. Since this expands as $\sum_{j=0}^{H-1}\kappa^{H-1-j}\left(n^{\frac{1}{H-1}}\right)^j$, it is clear that the expected total number of transmissions increases with $n$ (for fixed $c,k$ and $H$).
	\begin{figure} 
		\centering
		\subfloat[Minimum retransmission probability]{%
			\includegraphics[width=\linewidth]{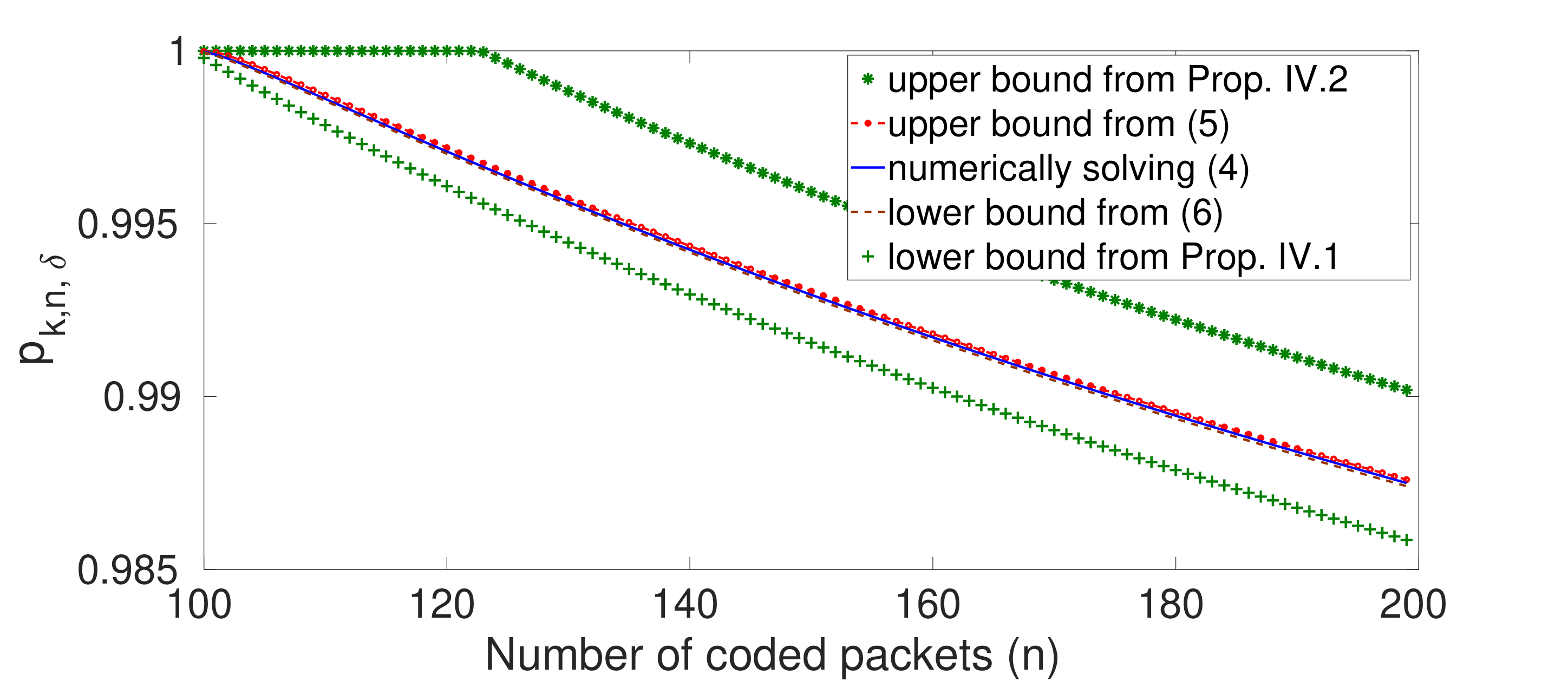}}
		\label{4a}
		\subfloat[Expected total number of transmissions]{%
			\includegraphics[width=\linewidth]{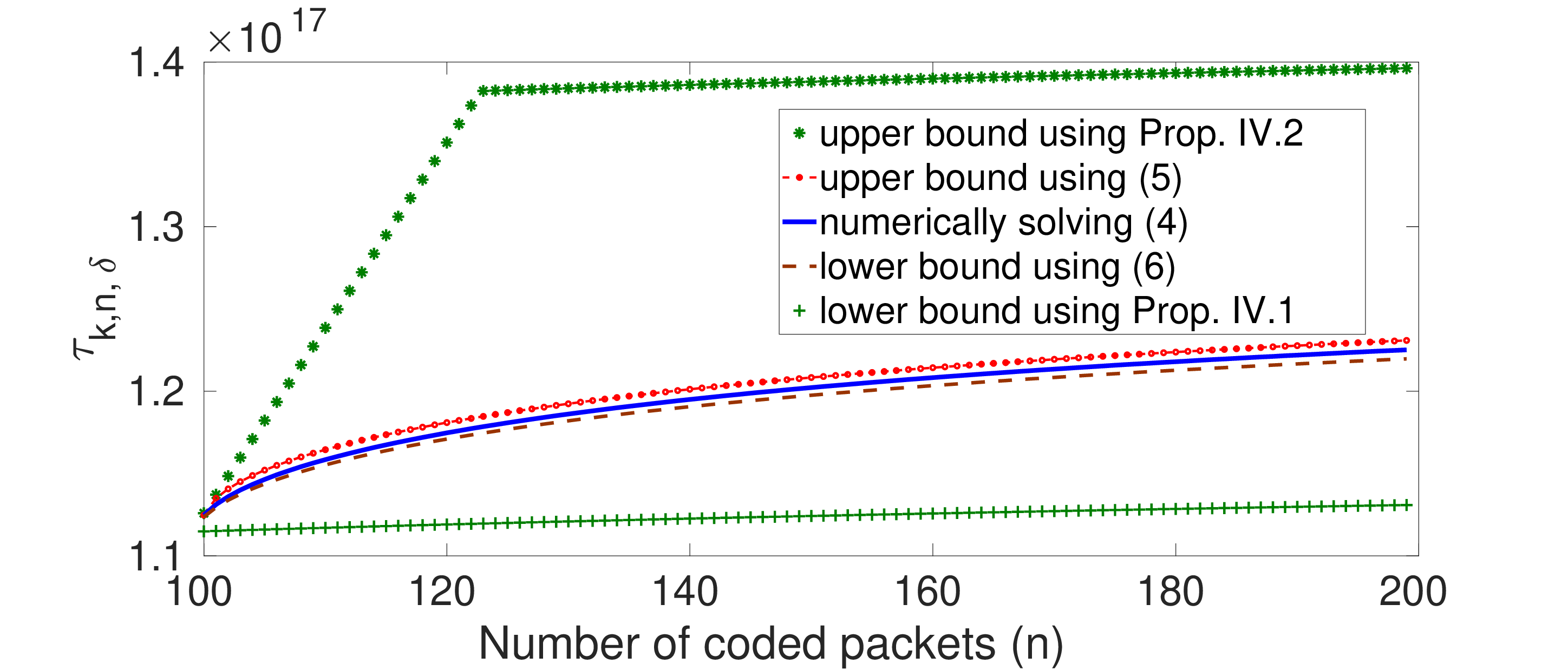}}
		\label{4b}\\
		\caption{The middle curves are plots of the true values of $p_{k,n,\delta}$ and $\tau_{k,n,\delta}$ obtained from (\ref{exprtomin}) and \eqref{eq:tau_tree}, for $k = 100$, $\delta = 0.1$ and $H=50$. The other curves are bounds obtained via Propositions \ref{prop:p_lobnd} and \ref{prop:p_upbnd},  \eqref{Eq:normal_tree_lower},   \eqref{Eq:normal_tree_upper} and \eqref{eq:tau_tree}.}
		\label{fig:tree_plots} 
	\end{figure}
	
In summary, introducing redundancy in the form of coding into the probabilistic retransmission protocol on a rooted binary tree (and more generally, on a rooted $d$-ary tree) is not beneficial in terms of the overall energy expenditure in the network.
	
\section{Grids}\label{sec:grid}

	\begin{figure}
		\centering
		\includegraphics[width=0.4\linewidth]{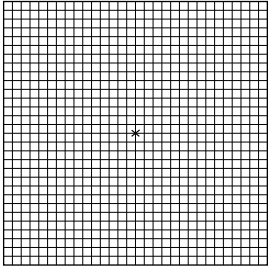}
		\caption{The source node ($\times$) is at the centre of the $31 \times 31$ grid.}
		\label{fig:grid31}
		\vspace{-1.5em}
	\end{figure}%
	
	Consider, for an odd integer $m > 1$, the $m \times m$ grid $\Gamma_m := {[-\frac{m-1}{2},\frac{m-1}{2}]}^2 \cap \Z^2$ centred at the origin. The source node is assumed to be at the centre of the grid. Simulation results for the probabilistic forwarding algorithm on the $31 \times 31$ grid (in Fig.~\ref{fig:grid31}) were presented in Fig. \ref{fig:grid_simu}. In this section, we try to explain these observations by developing an analysis that is at least valid for large $m$. Specifically, we turn to the theory of site percolation on the integer lattice $\Z^2$ to explain the $p_{k,n,\delta}$ and $\tau_{k,n,\delta}$ curves obtained via simulations on large grids $\G_m$.

\subsection{Site percolation on $\Z^2$} \label{sec:perc}
	
	We start with a brief description of the site percolation process (see e.g.\ \cite{grimmett}) on $\Z^2$. This is an i.i.d.\ process ${(X_u)}_{u\in\Z^2}$, with $X_u \sim \text{Ber}(p)$ for each $u \in \Z^2$, where the probability $p \in [0,1]$ is a parameter of the process. Let $\P_1$ denote the push-forward measure of the process on $\{0,1\}^{\Z^2}$ (or, in other words, the product measure $\otimes_{u} \nu_u$, with $\nu_u \sim \text{Ber}(p) \ \forall\, u \in \Z^2$). A node or \emph{site} $u \in \Z^2$ is \emph{open} if $X_u = 1$, and is \emph{closed} otherwise. For $u = (u_x,u_y) \in \Z^2$, define $|u| := |u_x|+|u_y|$. Two sites $u$ and $v$ are joined by an edge, denoted by $u$---$v$, iff $|u-v| = 1$. The next few definitions are made with respect to a given realization of the process ${(X_u)}_{u\in\Z^2}$. Two sites $u$ and $v$ are connected by an \emph{open path}, denoted by $u \longleftrightarrow v$, if there is a sequence of sites $u_0 = u, u_1,u_2,\ldots,u_n = v$ such that $u_k$ is open for all $k \in \{0,1,\ldots, n\}$ and $u_{k-1}$---$u_k$ for all $k \in [n]$. The \emph{open cluster}, $C_u$, containing the site $u$ is defined as $C_u=\{v\in \mathbb{Z}^2 | u \longleftrightarrow v\}$. Thus, $C_u$ consists of all sites connected to $u$ by open paths. In particular, $C_u = \emptyset$ if $u$ is itself closed. The \emph{boundary}, $\partial C_u$, of a non-empty open cluster $C_u$ is the set of all closed sites $v \in \Z^2$ such that \mbox{$v$---$w$} for some $w \in C_u$. The set $C_u^+ := C_u \cup \partial C_u$ is called an \emph{extended cluster}. The cluster $C_u$ (resp.\ $C_u^+$) is termed an \emph{infinite open cluster (IOC)} (resp.\ \emph{infinite extended cluster (IEC)}) if it has infinite cardinality. Note that $C_u^+$ is infinite iff $C_u$ is infinite.

It is well-known that there exists a \emph{critical probability} $p_c \in (0,1)$ such that for all $p < p_c$, there is almost surely (with respect to $\P_1$) no IOC, while for all $p > p_c$, there is almost surely a unique IOC. We do not know what happens at $p = p_c$, as the exact value of $p_c$ is itself not known (for site percolation on $\Z^2$). It is believed that $p_c \approx 0.59$ \cite[Chapter~1]{grimmett}. Another quantity of interest, which will play a crucial role in our analysis, is the \emph{percolation probability} $\theta(p)$, defined to be the probability that the origin $\0$ is in an IOC. In our analysis, we also consider the probability, $\theta^+(p)$, of the origin $\0$ being in an IEC. Clearly, from our definition of the IEC, for $p < p_c$, we have $\theta^+(p) = \theta(p) = 0$; for $p > p_c$, it is not difficult to see that $\theta^+(p) \ge \theta(p) > 0$. It is known that $\theta(p)$ is non-decreasing and infinitely differentiable in the region $p>p_c$ \cite{russo1978note}, but there is no analytical expression known for it. The following lemma, outlined in \cite{Shen2006dirbroad}, expresses $\theta^+(p)$ in terms of $\theta(p)$. 

\begin{lem}
For any $p > p_c$, we have $\theta^+(p)=\frac{\theta(p)}{p}$.
\label{lem:theta}
\end{lem} 
\begin{proof}
Let $C$ and $C^+$ be the (unique) IOC and IEC, respectively. We then have
\begin{equation}
\theta(p) = \P_1(\0 \in C) = \P_1(\0 \in C^+ \text{ and } \0 \text{ is open}).
\label{eq:theta}
\end{equation}
Now, observe that the event $\{\0 \in C^+\}$ is determined purely by the states of the nodes other than the origin. Hence, this event is independent of the event that $\0$ is open. Thus, the right-hand side (RHS) of \eqref{eq:theta} equals $\theta^+(p) \cdot p$, which proves the lemma.
\end{proof}

\begin{figure}
	\centering
	\includegraphics[width=0.5\textwidth]{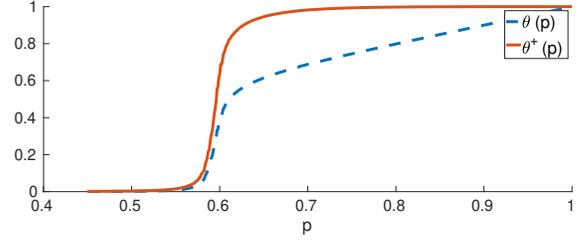}
	\caption{$\theta(p)$ and $\theta^+(p)$ vs. $p$}
	\label{fig:theta}
\end{figure}%

Fig.~\ref{fig:theta} plots $\theta(p)$ and $\theta^+(p)$ as functions of $p$, the former being obtained via simulations based on the theorem below. 

\begin{thm}\label{thm:theta}
Let $p>p_c$, and let $C$ and $C^+$, respectively, be the (almost surely) unique IOC and IEC of a site percolation process on $\mathbb{Z}^2$ with parameter $p$. Then, almost surely, we have
$$\lim_{m\rightarrow \infty}\frac{1}{m^2} |C\cap \G_m|=\theta(p) \ \text{ and } \ 
  \lim_{m\rightarrow \infty}\frac{1}{m^2}|C^+\cap \G_m|=\theta^+(p).$$
\end{thm}

The theorem is obtained as a straightforward application of an ergodic theorem for multi-dimensional i.i.d.\ random fields \cite[Proposition~8]{newman1981infinite} --- see Section~\ref{sec:ergthms} of the appendix. Using the dominated convergence theorem (DCT), we also have
\begin{align*}
\lim_{m\rightarrow \infty} \E\left[\frac{1}{m^2}|C\cap \G_m|\right]=\theta(p) &\hspace{1cm} \text{ and }\\
 \lim_{m\rightarrow \infty} \E\left[\frac{1}{m^2}|C^+\cap \G_m|\right]&=\theta^+(p).
\end{align*}

Based on the first equation above, to obtain an estimate of $\theta(p)$, the site percolation process with parameter $p$ was simulated on a $1001 \times 1001$ grid and the average fraction of nodes (averaged over $100$ realizations of the process) in the largest open cluster was taken to be the value of $\theta(p)$. These are the values of $\theta(p)$ plotted in Fig.~\ref{fig:theta}. We would like to emphasize that the plots in the figure should only be trusted for $p > p_c$, as Theorem~\ref{thm:theta} is only valid in that range. However, as the exact value of $p_c$ is unknown, simulation results are reported for the range of $p$ values shown in the plot.

\subsection{Relating site percolation to probabilistic forwarding} \label{sec:perc_to_fwding}
Site percolation on $\Z^2$ is a faithful model for probabilistic forwarding of a single packet on the infinite lattice $\Z^2$. The origin $\0$ is the source of the packet. The open cluster, $C_{\0}$, containing the origin $\0$ corresponds to the set of nodes that transmit (forward) the packet, and the extended cluster $C^+_{\0}$ corresponds to the set of nodes that receive the packet. The only caveat is that, since the source is assumed to always transmit the packet, we must consider only those realizations of the site percolation process in which the origin $\0$ is open. In other words, we must consider the site percolation process, conditioned on the event that the origin is open. By extension, the probabilistic forwarding of $n$ coded packets corresponds to $n$ independent site percolation processes on $\Z^2$, conditioned on the event that the origin is open in all $n$ percolations. 

Let O denote the event that the origin is open in all $n$ percolations. In our analysis, we will use $\P^\mathrm{o}$ and $\E^\mathrm{o}$, respectively, to denote the probability measure and expectation operator conditioned on the event O, and $\P$ and $\E$ for the unconditional versions of these. 


\subsection{Analysis of probabilistic forwarding on a large (finite) grid} \label{sec:grid_analysis}
In this section, we analyze the probabilistic forwarding mechanism on the finite grid $\G_m$ using the following approach. We map the probabilistic forwarding mechanism on $\G_m$ onto the probabilistic forwarding mechanism on the infinite $\Z^2$ lattice. From the discussion in the previous subsection, this is nothing but $n$ independent site percolations on $\Z^2$ conditioned on the event O. Using ergodic theorems for the site percolation process, we get a handle on the expected number of nodes that receive at least $k$ out of the $n$ packets from the origin on $\Z^2$. This, in turn, is used to obtain estimates of $p_{k,n,\delta}$ and $\tau_{k,n,\delta}$. In our analysis, we will assume that we operate in the super-critical region, i.e., $p>p_c$. We provide a justification for this assumption in Section \ref{sec:supercrit}. 

Denote by $R_{k,n}(\Gamma_m)$, the number of successful receivers in $\G_m$, i.e., the number of nodes that receive at least $k$ out of $n$ packets during the probabilistic forwarding mechanism on $\G_m$. The following theorem is our main result for grids. Its proof is quite technical, and is presented in the next subsection.

\begin{thm}
	For $p>p_c$, we have
	\begin{align*}
	\lim_{m\rightarrow\infty}  \E&\left[\frac{R_{k,n}(\Gamma_m)}{m^2} \right] \ = \ \\
	&\sum_{t=k}^{n} \sum_{j=k}^{t}\binom{n}{t}\binom{t}{j}(\theta^+(p))^{t+j}(1-\theta^+(p))^{n-j}.
	\end{align*}
	Equivalently,
	\begin{equation}
	\lim_{m\rightarrow\infty}  \E\left[\frac{R_{k,n}(\Gamma_m)}{m^2} \right] \ = \ \P(Y\ge k),
	\label{eq:pkndelta_grid}
	\end{equation}
	where $Y \sim \mathrm{Bin}(n,(\theta^+(p))^2)$.
	\label{thm:rkn}
\end{thm}

Thus, for $k,n,\delta$ fixed, we have for all sufficiently large grids $\G_m$,
\begin{equation}
p_{k,n,\delta}(\Gamma_m) \ \approx \ \inf \{ p \ | \ Pr(Y\ge k) \ge 1-\delta\}, \label{eq:pkndelta_simp}
\end{equation}
where $Y\sim \mathrm{Bin}(n,(\theta^+(p))^2)$. This can be evaluated numerically using the values of $\theta^+(p)$ plotted in Fig.~\ref{fig:theta}. For large $k$ and $n$, the probability $\P(Y \ge k)$ can be approximated well using the bounds given in Theorem~\ref{bin_normal_cdf} in the appendix. A sample of results thus obtained are shown in Fig.~\ref{fig:ergodic_pkndelta}. It is clear that these results match very well with those obtained from simulations on a $501 \times 501$ grid.

	\begin{figure}[t]
		\centering
		\includegraphics[width=0.5\textwidth]{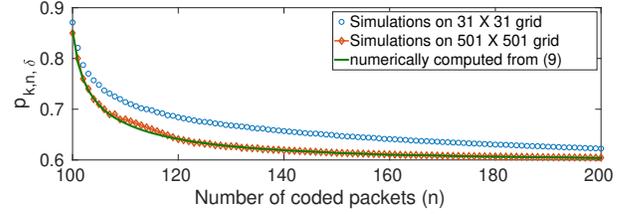}
		\caption{Comparison of the minimum forwarding probability obtained via simulations on a $31\times 31$ grid and a $501\times 501$ grid, with the results obtained numerically from \eqref{eq:pkndelta_simp}, for $k=100$ data packets and $\delta=0.1$.}
		\label{fig:ergodic_pkndelta}
		\vspace*{-1em}
	\end{figure}%
	\begin{figure}
		\centering
		\includegraphics[width=0.5\textwidth]{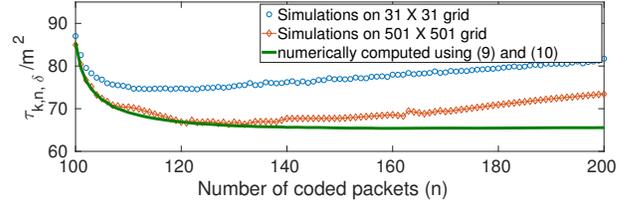}
		\caption{Comparison of the expected total number of transmissions, normalized by the grid size $m^2$, obtained via simulations on $\G_{31}$ and $\G_{501}$, with the expression from \eqref{eq:taukndelta}, for $k=100$ data packets and $\delta=0.1$.}
		\label{fig:ergodic_trans}
		\vspace*{-1em}
	\end{figure}%
We next look into estimating the expected total number of transmissions at a given forwarding probability $p$. Consider the transmission of a single packet on the finite grid $\G_m$. Let $T(\G_m)$ be the number of transmissions of the packet on the finite grid $\G_m$ and let $\mathcal{T}(\Z^2) \cap \G_m$ be the set of nodes in $\G_m$ which receive the packet from the origin and transmit it on the infinite $\Z^2$ lattice. It can be shown\footnote{This is shown using arguments entirely analogous to those used to show \eqref{eq:Rkn_Gm_Z2} in Section~\ref{sec:bigproof}. We omit the details.} that 
$$\lim_{m \rightarrow \infty}\frac{\E[T(\G_m)]}{m^2}=\lim_{m \rightarrow \infty}\frac{\E\left[\left|\mathcal{T}(\Z^2) \cap \G_m\right|\right]}{m^2}.$$
Now, $\mathcal{T}(\Z^2)$ is simply the open cluster $C_{\0}$ in the percolation framework. Thus, when normalized by the grid size $m^2$, the expected number of transmissions, $\E[T(\G_m)]$, for probabilistic forwarding on a large (but finite) grid $\G_m$ is well-approximated by $\E\bigl[|C_{\0} \cap \G_m| \ \big| \ \0 \text{ is open}\bigr]$. The following lemma gives an expression for this quantity in the limit as the grid size goes to infinity.

\begin{lem} For site percolation with $p > p_c$, we have
$$
 \lim_{m \to \infty} \frac{1}{m^2}\E\bigl[|C_{\0} \cap \G_m| \ \big| \ \0 \text{ is open}\bigr] \ = \ \frac{{\theta(p)}^2}{p}.
 $$
 \label{prop:C0}
 \end{lem}
 \begin{IEEEproof}
  We use $\P^{\0}$ and $\E^{\0}$, respectively, to denote the probability measure and expectation operator conditioned on the event that the origin $\0$ is open. Let $C$ be the (unique) IOC, and $A$ the event $\{\0 \in C\}$.  Then,
  \begin{align*}
  \lim_{m \to \infty} \E^{\0} & \left[\frac{1}{m^2} |C_{\0} \cap \G_m| \right] \\
  & = \lim_{m \to \infty} \E \left[\frac{1}{m^2} |C_{\0} \cap \G_m| \ \big| \ A \right] \P^{\0}(A) \\
   & \ \ \ \ \ \ \ \ \ \ +  \lim_{m \to \infty} \E^{\0} \left[\frac{1}{m^2} |C_{\0} \cap \G_m| \ \big| \ A^c \right] \P^{\0}(A^c) 
 \end{align*}
 
 Now, given $A^c$ (i.e., $\0 \notin C$), $C_{\0}$ is $\P^{\0}$-a.s.\ finite, and so by the DCT,
 ${\displaystyle \lim_{m \to \infty}} \E^{\0} \left[\frac{1}{m^2} |C_{\0} \cap \G_m| \ \big| \ A^c \right] = 0$. On the other hand, given $A$, we have $C_{\0} = C$. From Theorem~\ref{thm:theta}, we know that ${\displaystyle \lim_{m \to \infty}} \frac{1}{m^2} |C \cap \G_m| = \theta(p) \ \, \P_1$-a.s.. Moreover, this statement holds even when the probability measure $\P_1$ is conditioned on $A$, since $\P_1(A) = \theta(p) > 0$ for $p > p_c$. So, again by the DCT, ${\displaystyle \lim_{m \to \infty}} \E[\frac{1}{m^2} |C \cap \G_m| \mid A] = \theta(p)$. We have thus shown that
 $$
 \lim_{m \to \infty} \E^{\0} \left[\frac{1}{m^2} |C_{\0} \cap \G_m| \right] = \theta(p) \, \P^{\0}(A).
 $$
 
 The proof is completed by observing that $\P^{\0}(A) = \frac{\P_1(A)}{\P_1({\0} \text{ is open})} = \frac{\theta(p)}{p}$.
 \end{IEEEproof}
 \ \\
 Thus, in probabilistic forwarding of a single packet on a large grid $\G_m$, the expected number of transmissions, normalized by the grid size $m^2$, is approximately $\frac{{\theta(p)}^2}{p}$. Hence, when we have $n$ coded packets, by linearity of expectation, 
 the expected total number of transmissions, again normalized by the grid size $m^2$, is approximately $n \, \frac{{\theta(p)}^2}{p}$. In particular, setting $p = p_{k,n,\delta}$, we obtain
	\begin{equation}\label{eq:taukndelta}
	\frac{1}{m^2} \, \tau_{k,n,\delta}(\G_m) \ \approx \ n \frac{{\theta(p_{k,n,\delta})}^2}{p_{k,n,\delta}},
	\end{equation}
provided that $p_{k,n,\delta} > p_c$.

Fig.~\ref{fig:ergodic_trans} compares, for $k=100$ data packets and $\delta = 0.1$, the values of $\frac{1}{m^2} \tau_{k,n,\delta}$ obtained using \eqref{eq:taukndelta}, \eqref{eq:pkndelta_simp} and the $\theta(p)$ values from Fig.~\ref{fig:theta}, with those obtained via simulations on the $\G_{31}$ and $\G_{501}$ grids. The curve based on \eqref{eq:taukndelta}, \eqref{eq:pkndelta_simp} and $\theta(p)$ initially tracks the $\G_{501} $ curve well, but trails off after $n = 130$. This is because the former curve uses the approximation for $p_{k,n,\delta}$ in \eqref{eq:pkndelta_simp}, which, for any given $n$, is valid only for sufficiently large $m$. For values of $n$ larger than $130$, $m=501$ may not fall in the ``sufficiently large'' range. This is discussed in more detail in Section~\ref{sec:sufflargem}.

Nonetheless, it is instructive to note that, for fixed values of $k$ and $\delta$, the expression on the right-hand side (RHS) of \eqref{eq:taukndelta} is indeed minimized for some $n$. This can be verified numerically by plotting the RHS of \eqref{eq:taukndelta} using the values of $\theta(p)$ from Fig.~\ref{fig:theta} and the approximation to $p_{k,n,\delta}$ in \eqref{eq:pkndelta_simp}. Plots for $k=100$ and $\delta=0.1$ are shown in Fig.~\ref{fig:ergodic_bounds}. Observe that the curve plotted in Fig.~\ref{fig:ergodic_bounds}(b) is decreasing in $n$ till $n \approx 180$, and it increases thereafter, albeit very slowly. This indicates that, for $k=100$ and $\delta = 0.1$, the expected number of transmissions $\tau_{k,n,\delta}(\G_m)$ is minimized at $n \approx 180$ for all sufficiently large grids $\G_m$. Thus, our analysis provides theoretical validation, at least for large grids, for the observed behaviour of $\tau_{k,n,\delta}$ as a function of $n$, and indicates a benefit to introducing some coding into the probabilistic forwarding mechanism on grids.
	\begin{figure} 
		\centering
		\subfloat[Minimum retransmission probability]{%
			\includegraphics[width=\linewidth]{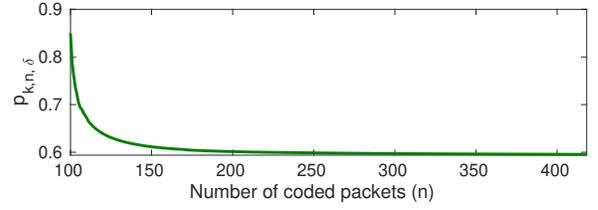}}
		\label{fig:ergodic_prob_bounds}
		\subfloat[Expected total number of transmissions normalized by the grid size $m^2$.]{%
			\includegraphics[width=\linewidth]{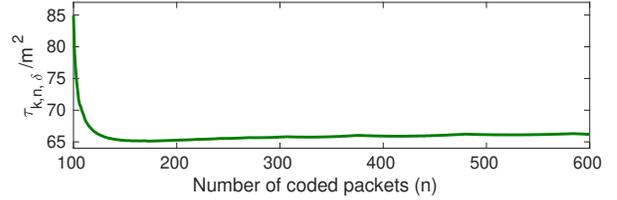}}
		\label{fig:ergodic_trans_bounds}\\
		\caption{The minimum forwarding probability is numerically computed from \eqref{eq:pkndelta_simp} and the expected number of transmissions is obtained via \eqref{eq:taukndelta}, for $k=100$ data packets and $\delta=0.1$.}
		\label{fig:ergodic_bounds} 
	\end{figure}

\subsection{Proof of Theorem~\ref{thm:rkn}} \label{sec:bigproof}
Let $\cR_{k,n}(\Z^2)$ denote the set of all nodes that receive at least $k$ of the $n$ coded packets during the probabilistic forwarding protocol on $\Z^2$. As a first step, we will show that $R_{k,n}(\Gamma_m)$ and $|\cR_{k,n}(\Z^2) \cap \G_m|$ are the same in expectation, in the limit as the grid size, $m$, goes to infinity. In general, it is only true that $R_{k,n}(\Gamma_m)$ is stochastically dominated\footnote{A random variable $X$ is stochastically dominated by a random variable $Y$ if $\P(X \ge x) \le \P(Y \ge x)$ for all $x \in \R$. For non-negative random variables, this implies that $\E[X] \le \E[Y]$.} by $|\cR_{k,n}(\Z^2) \cap \G_m|$, since a node in $\cR_{k,n}(\Z^2) \cap \G_m$ could receive packets from the origin through paths in $\Z^2$ that do not lie entirely within $\G_m$.

In the percolation jargon (on $\Z^2$), $\cR_{k,n}(\Z^2) \cap \G_m$ comprises those nodes of $\G_m$ that are in the extended cluster containing the origin ($C_{\0}^+$) in at least $k$ out of $n$ percolations. Recall that a node $u$ is in $C_{\0}^+$ if either the node $u$ or some one-hop neighbour of $u$ is connected to the origin through an open path. Call such an open path a \emph{conduit} (for a packet) from the origin to $u$. If a conduit lies completely within $\G_m$, we call it a \emph{$\G_m$-conduit}. We also say that, if vertex $u$ has a conduit, it is necessarily in $C_{\0}^+$.

The nodes in $\cR_{k,n}(\Z^2) \cap \G_m$ may have received some packets from the origin through $\G_m$-conduits, and some others through conduits that go outside $\G_m$. We are interested in the former, since, when operating on a finite grid $\G_m$, nodes of  $\cR_{k,n}(\Z^2) \cap \G_m$ without $\G_m$-conduits cannot be successful receivers in $\G_m$. More precisely, we are interested in those nodes\ of $\G_m$ which are part of the extended cluster containing the origin through at least one $\G_m$-conduit, in at least $k$ out of the $n$ percolations. Note that these are the nodes that receive at least $k$ out of the $n$ packets in the finite grid model; we denote this collection of nodes by $\cR_{k,n}(\Gamma_m)$. Thus, $|\cR_{k,n}(\Gamma_m)|=R_{k,n}(\Gamma_m)$. We denote the remaining nodes by $\overline{\cR}_{k,n}(\G_m) := (\cR_{k,n}(\Z^2) \cap \G_m) \backslash \cR_{k,n}(\Gamma_m)$. Thus, $\cR_{k,n}(\Gamma_m)$ and $\overline{\cR}_{k,n}(\G_m) $ form a partition of $\cR_{k,n}(\Z^2) \cap \G_m $, i.e.,
	\begin{align}
	\cR_{k,n}(\Gamma_m) &\cap \overline{\cR}_{k,n}(\G_m) =\emptyset
\nonumber \\
	&\text{and} \nonumber \\
	\cR_{k,n}(\Gamma_m) &\cup \overline{\cR}_{k,n}(\G_m) =\cR_{k,n}(\Z^2) \cap \G_m. 
	\label{eq:pathunion}
	\end{align}
 Note that any node in $\overline{\cR}_{k,n}(\G_m)$ has the property that for at least one of the packets it receives, any conduit through which it receives that packet \textit{necessarily} goes outside $\G_m$. Such a node is said to \textit{receive at least one packet from outside $\G_m$}. It does not receive this packet through any $\G_m$-conduit.

We first show that the expected fraction of nodes in $\G_m$ that receive at least one packet from outside $\G_m$ vanishes asymptotically with the grid size $m$. In this direction, we will need the following definition:
For $0< \epsilon < 4$, let
\begin{align*}
\Gamma_{m,\epsilon} :=\left\{\begin{aligned}
\G_{\left\lfloor m\sqrt{1-\frac{\epsilon}{4}}\right\rfloor}, \hspace{0.8cm} &\text{ if }\left\lfloor m\sqrt{1-\frac{\epsilon}{4}}\right\rfloor \text{ is odd}\\
\G_{\left\lfloor m\sqrt{1-\frac{\epsilon}{4}}\right\rfloor-1}, \hspace{0.8cm} &\text{ if }\left\lfloor m\sqrt{1-\frac{\epsilon}{4}}\right\rfloor \text{ is even}
\end{aligned}\right\}
\end{align*}
Recall that $\G_m$ was defined as 
$\Gamma_m := {[-\frac{m-1}{2},\frac{m-1}{2}]}^2 \cap \Z^2$ when $m$ was odd. We will think of $\Gamma_{m,\epsilon}$ as being $\G_{ m\sqrt{1-\frac{\epsilon}{4}}}$ in our calculations, and hence the number of nodes in $\G_{m,\epsilon}$ is approximately $m^2\left(1-\frac{\epsilon}{4}\right)$. 
 
\begin{lem}
	Let $p_c$ be the critical probability for site percolation. For $p>p_c$, we have
	\begin{equation*}
	\lim\limits_{m\rightarrow \infty}\frac{1}{m^2}\E^\mathrm{o}\left[|\overline{\cR}_{k,n}(\G_m)| \right]=0
	\end{equation*}
	\label{lem:outpath}
\end{lem}
\begin{IEEEproof}
	Fix an $\epsilon>0$. We will find an $m_0$ such that $\frac{1}{m^2}\E^\mathrm{o}\left[\overline{\cR}_{k,n}(\G_m)\right]<\epsilon$ for all $m\ge m_0$. This will prove the lemma.
	\par 
	Any node in $\overline{\cR}_{k,n}(\G_m)$ has a conduit in at least $k$ out of the $n$ packet transmissions on $\Z^2$ and receives at least one packet from outside $\G_m$.
	Denote by $M_j$  the event that node $j$ receives at least one of the $n$  packets from outside $\G_m$. Recall that this means that node $j$ does not have any $\G_m$-conduit for this packet. We then have,
	\begin{align*}
	\E^{\mathrm{o}}  \left[\frac{| \overline{\cR}_{k,n}(\G_m) |}{m^2} \right] &\le \E^{\mathrm{o}}  \left[\frac{1}{m^2}\sum_{j\in\G_m} \mathds{1}_{M_j}  \right], \\
	&= \E^{\mathrm{o}}  \left[\frac{1}{m^2}\sum_{j\in\G_{m,\epsilon} } \mathds{1}_{M_j} \right] \\
	&\hspace{1.5cm}+ \E^{\mathrm{o}}  \left[\frac{1}{m^2}\sum_{j\in\G_m \backslash \G_{m,\epsilon}} \mathds{1}_{M_j} \right],
	\end{align*}
	where $\mathds{1}_{M_j}$ is the indicator random variable for the event $M_j$, i.e., $\mathds{1}_{M_j} =1$ if $M_j$ occurs, and $\mathds{1}_{M_j} = 0$ otherwise. Since there are $m^2-m^2\left(1-\frac{\epsilon}{4}\right)=\frac{m^2\epsilon}{4}$ nodes in $\G_m \backslash \G_{m,\epsilon}$, the latter term can be further bounded to obtain,
	\begin{align}	
	\E^{\mathrm{o}}  \left[\frac{| \overline{\cR}_{k,n}(\G_m) |}{m^2} \right] &\le \frac{1}{m^2}\sum_{j\in\G_{m,\epsilon}}\P^{\mathrm{o}}\left( M_j \right) +\frac{\epsilon}{4}.
	\label{Eq:outrec}
	\end{align}

	The summation above can be split over those nodes which are on the boundary of $\G_{m,\epsilon}$ and those in the interior. The former term contains at most $4m\sqrt{1-\epsilon/4}$ nodes. The latter term involves those nodes which receive at least one packet from outside $\G_m$. Hence, in at least one percolation, such nodes have a path from the origin as shown in Fig. \ref{Fig:annulus}. This, then implies that there cannot be an open loop in the annulus $\G_m \backslash \G_{m,\epsilon}$ as indicated by the dotted line in Fig. \ref{Fig:annulus}.  Let $K_m$ be the event that there is no open loop around the origin in the annulus $\G_m \backslash \G_{m,\epsilon}$ in at least one percolation. 
	\begin{figure}
		\centering
		\includegraphics[width=0.35\textwidth]{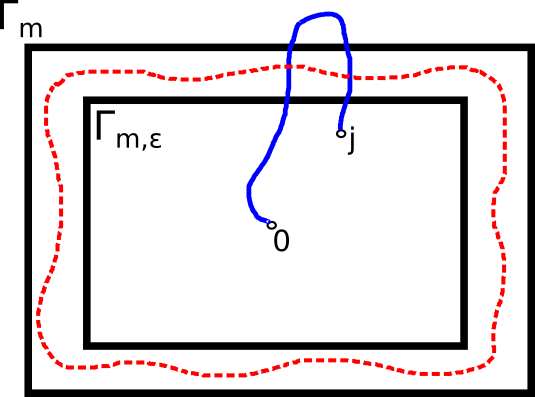}
		\caption{Illustration of open loop in the annulus $\G_m \backslash \G_{m,\epsilon}$. Here the vertex $j$ receives the packet from origin $0$, only along the path that is depicted.}
		\label{Fig:annulus}
	\end{figure}
	We then obtain,
	\begin{align}
	\frac{1}{m^2}\sum_{j\in\G_{m,\epsilon} }&\P^{\mathrm{o}}\left( M_j \right) \nonumber  \\
	&\le \frac{1}{m^2}\left[4m\sqrt{1-\frac{\epsilon}{4}} \ \right]+\left(1-\frac{\epsilon}{4}\right)\P^{\mathrm{o}}\left( K_m \right)\nonumber  \\
	&= \frac{4}{m}\sqrt{1-\frac{\epsilon}{4}}+\left(1-\frac{\epsilon}{4}\right)\left(1-\P^{\mathrm{o}}\left( K_m^c \right)\right).
	\label{Eq:annul}
	\end{align}
	The event $K_m^c$ is the event that there is an open loop in the annulus $\G_m \backslash \G_{m,\epsilon}$ in each of the $n$ percolations. Note that this is an increasing event and so is the event $\mathrm{O}$. Using the FKG inequality (see \cite[Chapter 2]{grimmett}), we have that
	\begin{align}
	\P^{\mathrm{o}}\left(K_m^c\right)\ &\ \ =\ \  \frac{\P\left(K_m^c \cap \mathrm{O}\right)}{\P\left(\mathrm{O}\right)}, \nonumber \\
	&\stackrel{(FKG)}{\ge}\frac{\P\left( K_m^c \right)\P\left( \mathrm{O}\right)}{\P\left(\mathrm{O}\right)},\nonumber  \\
	&\ \ =\ \ \P\left( K_m^c \right).
	\label{eq: fkg_trick}
	\end{align}
	On $\{0,1\}^{\mathbb{Z}^2}$, define $\mathsf{Ann}$ to be the event that there is an open loop in the annulus $\G_m \backslash \G_{m,\epsilon}$. Exploiting the independence of packet transmissions, we have that $\P\left( K_m^c \right)=\P_1\left( \mathsf{Ann} \right)^n$. Substituting \eqref{eq: fkg_trick} and \eqref{Eq:annul} in  \eqref{Eq:outrec}, and using this result, we obtain,
	\begin{align*}
	\E^{\mathrm{o}}  &\left[\frac{| \overline{\cR}_{k,n}(\G_m) |}{m^2} \right] \\
	&\le  \frac{4}{m}\sqrt{1-\frac{\epsilon}{4}}+\left(1-\frac{\epsilon}{4}\right)\left(1-\P_1\left( \mathsf{Ann} \right)^n\right)+\frac{\epsilon}{4}.
	\end{align*}
	For super-critical site percolation process on $\Z^2$ and a fixed $\epsilon>0$, the probability of an open loop in the annulus $\G_m \backslash \G_{m,\epsilon}$ around the origin is known to approach $1$ as $m\rightarrow \infty$ (see \cite{grimmett} for an idea of the proof, and \cite{russo1978note} for specific results for site percolation) i.e. $\P_1(\mathsf{Ann})\rightarrow 1$ as $m\rightarrow \infty$. Thus we can find an $m_0$ such that each of the first two terms on the RHS in the above expression are less than $\frac{\epsilon}{4}$ for all $m\ge m_0$. This is the required $m_0$. 
	\end{IEEEproof}

 
 \par Since $\cR_{k,n}(\Z^2) \cap \G_m$ is a disjoint union of nodes in $\cR_{k,n}(\Gamma_m)$ and $\overline{\cR}_{k,n}(\G_m)$, the previous lemma shows that 
 \begin{equation}
 \lim_{m \to \infty} \frac{1}{m^2} \E^\mathrm{o}\big[|\cR_{k,n}(\G_m)|\big] = \lim_{m \to \infty} \frac{1}{m^2} \E^\mathrm{o}\big[|\cR_{k,n}(\Z^2) \cap \G_m|\big],
 \label{eq:Rkn_Gm_Z2}
 \end{equation}
This provides us with a mapping between the probabilistic forwarding mechanism on a large (but finite) grid $\G_m$ and the infinite lattice $\Z^2$.

	In our analysis on the grid, we will be interested in the expected value of $|\cR_{k,n}(\Z^2) \cap \G_m) |$ when conditioned on the event $A_T^+$, defined, for any $T\subset [n]$, as the event that the origin	is in the IEC in exactly the percolations indexed by $T$. As a corollary of Lemma \ref{lem:outpath}, we also obtain
\begin{cor}
	Let $p_c$ be the critical probability for site percolation. For $p>p_c$, we have
	\begin{equation*}
	\lim\limits_{m\rightarrow \infty}\frac{1}{m^2}\E^\mathrm{o}\left[\overline{\cR}_{k,n}(\G_m) \big| A_T^+\right]=0.
	\end{equation*}
	\label{cor:outpath}
\end{cor}
\begin{IEEEproof}
	The proof is along similar lines as that of Lemma \ref{lem:outpath} but with additional conditioning on the event $A_T^+$. More specifically,  \eqref{Eq:annul} would have $	\P^{\mathrm{o}}\left\{ K_m^c \ \big| \ A_T^+\right\}$ on the RHS. Notice that $A_T^+$ is an increasing event and hence $\mathrm{O} \cap A_T^+$ is also increasing. Thus,
	\begin{align}
	\P^{\mathrm{o}}\left( K_m^c \ \big| \ A_T^+\right)&= \frac{\P\left( K_m^c \cap A_T^+ \cap \mathrm{O}\right)}{\P\left(A_T^+ \cap \mathrm{O}\right)} \nonumber \\
	&\stackrel{(FKG)}{\ge}\frac{\P\left( K_m^c \right)\P\left( A_T^+ \cap  \mathrm{O}\right)}{\P\left( A_T^+ \cap \mathrm{O}\right)}\nonumber  \\
	&=\P\left( K_m^c \right).
	\label{eq: fkg_trick_cor}
	\end{align}
	Using this in  \eqref{Eq:annul} and following subsequent steps from the lemma, we get the statement of the corollary.
\end{IEEEproof}

It is to be justified that such conditioning can indeed be done, i.e., the event $A_T^+$ has a positive probability for the specified range of values of $p$. The following proposition relates the probability of the event $A^+_T$, conditioned on the event that the origin is open in all $n$ percolations, to $\theta^+(p)$.
\begin{prop} For any $T \subseteq [n]$ with $|T| = t$, we have
	$$\P^{\mathrm{o}}(A^+_T)\ =\ (\theta^+(p))^{t}(1-\theta^+(p))^{n-t}.$$
	\label{prop:formulap}
\end{prop}
\begin{IEEEproof}
By definition, $\P^{\mathrm{o}}(A_T^+) = \P(A_T^+ \mid \mathrm{O})$. Note that, in a given percolation, conditioned on $\0$ being open, the event $\{\0 \text{ is in the IEC}\}$ is the same as the event $\{\0 \text{ is in the IOC}\}$. Consequently, conditioned on O, the event $A_T^+$ is the same as the event, $A_T$, that the origin is in the IOC in exactly the percolations indexed by $T$. Hence, 
$$\P^{\mathrm{o}}(A_T^+) = \P(A_T \mid \mathrm{O}) = \frac{\P(A_T \cap \mathrm{O})}{\P(\mathrm{O})}.$$
The denominator equals $p^n$. The numerator is the event that the origin is in the IOC in exactly the percolations indexed by $T$, and is open but in a finite cluster in the remaining $n-|T|$ percolations. In a given percolation, the probability that the origin is open but in a finite cluster is $p-\theta(p)$. Thus, we have $\P(A_T \cap \mathrm{O})= (\theta(p))^{|T|}(p-\theta(p))^{n-|T|}$. The result now follows from the fact (Lemma~\ref{lem:theta}) that $\theta^+(p) = \frac{\theta(p)}{p}$.
\end{IEEEproof}


Since $\theta^+(p)>0$ for $p>p_c$, we have that $\P^{\mathrm{o}}(A^+_T)>0$ as well.

We now state an ergodic theorem for $n$ independent copies of the site percolation process on $\Z^2$, which will aid us in analyzing $\left|\cR_{k,n}(\Z^2) \cap \G_m\right|$. For this, let $C_{k,n}^+$ be the set of all sites in $\Z^2$ that belong to the IEC in at least $k$ out of $n$ independent percolations. By a simple application of standard ergodic theorems as detailed in Section \ref{sec:ergthms} of the appendix, we have the following theorem.
\begin{thm} We have
	$$\lim_{m \rightarrow \infty}\frac{1}{m^2} |C_{k,n}^+ \cap \G_m| = \theta^+_{k,n}(p) \ \ \ \ \ \ \P\text{-a.s.}$$
	where 
	$$
	\theta^+_{k,n}(p) =  \sum_{j = k}^n \binom{n}{j} (\theta^+(p))^j (1-\theta^+(p))^{n-j}
	$$
	is the probability that the origin belongs to the IEC in at least $k$ out of the $n$ percolations.
	\label{thm:prodtheta}
\end{thm} 
From the theorem, we derive a useful fact that plays a key role in our analysis. Since the event, say $A_n$, that the origin is in the IOC in all $n$ percolations has positive probability ($\theta(p)^n > 0$ for $p > p_c$), the theorem statement also holds almost surely when conditioned on $A_n$. Hence, by the DCT, we also have

\begin{cor}
\begin{equation*}
\lim_{m \rightarrow \infty} \E\left[\frac{1}{m^2} |C_{k,n}^+ \cap \G_m| \ \bigg| \ A_n \right] = \theta^+_{k,n}(p) \, .
\end{equation*}
\label{cor:EgivenAn}
\end{cor}

We are now in a position to prove Theorem~\ref{thm:rkn}, which is restated below for convenience. The proof is obtained by carefully relating $\cR_{k,n}(\Z^2)$ to the set $C_{k,n}^+$, and then using Corollary \ref{cor:EgivenAn}.

\begin{thm}[Restatement of Theorem~\ref{thm:rkn}]
	For $p>p_c$, we have
	\begin{align*}
	\lim_{m\rightarrow\infty}  \E^{\mathrm{o}}&\left[\frac{|\cR_{k,n}(\Gamma_m)|}{m^2} \right] \ = \ \\
	&\sum_{t=k}^{n} \sum_{j=k}^{t}\binom{n}{t}\binom{t}{j}(\theta^+(p))^{t+j}(1-\theta^+(p))^{n-j}.
	\end{align*}
	Equivalently,
	\begin{equation}
	\lim_{m\rightarrow\infty}  \E^{\mathrm{o}}\left[\frac{|\cR_{k,n}(\Gamma_m)|}{m^2} \right] \ = \ \P(Y\ge k),
	\label{eq:pkndelta_grid}
	\end{equation}
	where $Y \sim \mathrm{Bin}(n,(\theta^+(p))^2)$.
	\label{thm:rkn-repeat}
\end{thm}
\begin{IEEEproof}
	Before we begin, recall from \eqref{eq:pathunion} that $\cR_{k,n}(\Gamma_m)$ and $\overline{\cR}_{k,n}(\G_m) $ form a partition of $\cR_{k,n}(\Z^2) \cap \G_m $.	In the framework of $n$ independent site percolations, $\cR_{k,n}(\Z^2)$ is the set of sites in $\Z^2$ that are in the extended cluster containing the origin in at least $k$ of the $n$ percolations (conditioned on the origin being open). 
	
	We start with
	\begin{align}
	\E^{\mathrm{o}} \left[|\cR_{k,n}(\Gamma_m)|\right] & =  \ \sum_{t = 0}^n \sum_{T \subseteq [n]: \atop |T|=t}  \E^{\mathrm{o}}  \left[|\cR_{k,n}(\Gamma_m)| \ \big| \ A_T^+\right] \, \P^{\mathrm{o}}(A_T^+).
	\label{eq:Eo_sum}
	\end{align}
 	Our approach in the ensuing discussion would be to first obtain results for $\cR_{k,n}(\Z^2) \cap \G_m$, and then transfer them to $\cR_{k,n}(\Gamma_m)$. Motivated by our discussion following Lemma \ref{lem:outpath}, consider the summand of  \eqref{eq:Eo_sum} with $\cR_{k,n}(\Gamma_m)$ replaced by $\cR_{k,n}(\Z^2) \cap \G_m$, i.e., $\E^{\mathrm{o}}  \left[|\cR_{k,n}(\Z^2) \cap \G_m| \ \big| \ A_T^+\right]$. 
	\par Suppose that $|T|=t<k$. Given $A_T^+$, the origin is in the IEC in no more than $k-1$ of the percolations; hence, each site in $\cR_{k,n}(\Z^2)$ must belong to the finite cluster, denoted by $C_{\0}[j]$, in the $j$th percolation, for some $j \notin T$. As a result, given $A_T^+$, $\cR_{k,n}(\Z^2)$ is contained in the union $\cup_{j \notin T} C_{\0}[j]$, which is finite $\P^{\mathrm{o}}$-a.s, so that ${\displaystyle \lim_{m \to \infty}} \frac{1}{m^2} |\cR_{k,n}(\Z^2) \cap \G_m| = 0$ \ $\P^{\mathrm{o}}$-a.s.. Since $\cR_{k,n}(\Gamma_m) \subseteq \cR_{k,n}(\Z^2) \cap \G_m$, we also obtain ${\displaystyle \lim_{m \to \infty}} \frac{R_{k,n}(\Gamma_m)}{m^2} = 0$ \ $\P^{\mathrm{o}}$-a.s.. Consequently, by the DCT, we have for any $T \subseteq[n]$ with $|T| < k$,
	\begin{align}
	{\displaystyle \lim_{m \to \infty}} \E^{\mathrm{o}}  &\left[\frac{1}{m^2}|\cR_{k,n}(\Z^2) \cap \G_m| \ \bigg| \ A_T^+\right] = 0 \hspace{1cm} \nonumber\\
	& \text{ and } \hspace{1cm} \lim_{m \to \infty} \E^{\mathrm{o}}  \left[\frac{|\cR_{k,n}(\Gamma_m)|}{m^2} \bigg| \ A_T^+\right]=0.
	\label{eq:t<k}
	\end{align}
	
	Next, consider any summand in \eqref{eq:Eo_sum} with $|T| = t \ge k$  and $\cR_{k,n}(\Gamma_m)$ replaced by $\cR_{k,n}(\Z^2) \cap \G_m$ as before. The sites in $\cR_{k,n}(\Z^2)$ can be exactly one of two types: those that belong to the extended cluster $C_{\0}^+$ in at least $k$ of the percolations indexed by $T$; and those that do not. Let $\cR_{k,T}$ be the subset of $\cR_{k,n}(\Z^2)$ consisting of sites of the first type, and let $\mathcal{Q} = \cR_{k,n}(\Z^2) \setminus \cR_{k,T}$. Thus,  
	\begin{align}
	\E^{\mathrm{o}} & \left[|\cR_{k,n}(\Z^2) \cap \G_m| \ \big| \ A_T^+\right]  \notag \\ 
	& \!\!\!\!\! = \ \E^{\mathrm{o}}\left[|\cR_{k,T} \cap \G_m| \ \big| \ A_T^+\right] + \E^{\mathrm{o}}\left[|\cQ \cap \G_m| \ \big| \ A_T^+\right].
	\label{eq:EoRkninfty}
	\end{align}
	
	Note that any site in $\cQ$ must belong to $C_{\0}^+$ in at least one percolation outside of $T$. In particular, given $A_T^+$, $\cQ$ is $\P^{\mathrm{o}}$-a.s.\ finite. Thus, arguing as in the $|T| < k$ case, we have 
	\begin{align}
	\lim_{m \to \infty} \E^{\mathrm{o}}  &\left[\frac{1}{m^2}| \cQ \cap \G_m| \ \big| \ A_T^+\right] = 0 \hspace{1cm} \nonumber\\
	&\text{and} \hspace{1cm} \lim_{m \to \infty} \E^{\mathrm{o}}  \left[\frac{|\cR_{k,n}(\Gamma_m) \cap \cQ|}{m^2} \ \bigg| \ A_T^+\right] = 0.
	\label{eq:Q}
	\end{align}
	
	Finally, note that 
	\begin{align*}
	\E^{\mathrm{o}}\left[|\cR_{k,T} \cap \G_m| \ \big| \ A_T^+\right] 
	& = \E\left[|\cR_{k,T} \cap \G_m| \ \big| \ A_T^+ \cap \mathrm{O}\right] \\
	& \stackrel{(a)}{=} \E\left[|C_{k,T}^+ \cap \G_m| \ \big| \ A_T^+ \cap \mathrm{O} \right] \\
	& \stackrel{(b)}{=} \E\left[|C_{k,T}^+ \cap \G_m| \ \big| \ A_T\right],
	\end{align*}
	where $A_T$ is the event that $\0$ is in the IOC in exactly the percolations indexed by $T$, and $C_{k,T}^+$ is the set of sites of $\Z^2$ that belong to the IEC in at least $k$ of the percolations indexed by $T$. The equality labeled~(a) above is due to the fact that, conditioned on $A_T^+ \cap O$, $\cR_{k,T} = C_{k,T}^+$. The equality labeled~(b) is because $A_T^+ \cap \mathrm{O} = A_T \cap \mathrm{O}$, and moreover, the event that $\0$ is open in the percolations outside $T$ is independent of the percolations indexed by $T$. 
	
	Thus, restricting our attention to \emph{only} the percolations indexed by $T$, we can apply Corollary \eqref{cor:EgivenAn} with $n = t$ to obtain ${\displaystyle \lim_{m \to \infty}} \E\left[\frac{1}{m^2}|C_{k,T}^+ \cap \G_m| \ \big| \ A_T\right] = \theta_{k,t}^+(p)$. Hence,
	\begin{equation}
	\lim_{m \to \infty} \E^{\mathrm{o}}  \left[\frac{1}{m^2}| \cR_{k,T}  \cap \G_m| \ \big| \ A_T^+\right] =  \theta_{k,t}^+(p).
	\label{eq:Eo:last}
	\end{equation}
	Now using \eqref{eq:pathunion} and the fact that $\cR_{k,T} \subset \cR_{k,n}(\Z^2)$, we obtain 
	\begin{align*}
	\cR_{k,T} \cap \G_m &= \cR_{k,T} \cap \cR_{k,n}(\Z^2) \cap \G_m \\
	&= \cR_{k,T}  \cap (\cR_{k,n}(\Gamma_m) \cup \overline{\cR}_{k,n}(\G_m) ) \\
	&= ( \cR_{k,T} \cap \cR_{k,n}(\Gamma_m)) \cup ( \cR_{k,T} \cap \overline{\cR}_{k,n}(\G_m)),
	\end{align*}
	in which the two sets $\cR_{k,T} \cap \cR_{k,n}(\Gamma_m)$  and $\cR_{k,T} \cap \overline{\cR}_{k,n}(\G_m)$ on the RHS are disjoint (from \eqref{eq:pathunion}). Using this, we can write the expectation term in \eqref{eq:Eo:last} as follows
	\begin{align}
	\E^{\mathrm{o}} &\left[ \frac{1}{m^2}| \cR_{k,T}  \cap \G_m| \ \big| \ A_T^+ \right] = \nonumber\\ 
	&\hspace{1cm} \E^{\mathrm{o}}  \left[\frac{1}{m^2}| \cR_{k,n}(\Gamma_m) \cap \cR_{k,T} | \ \bigg| \ A_T^+\right]+ \nonumber\\
	&\hspace{2cm} \E^{\mathrm{o}}  \left[\frac{1}{m^2}| \overline{\cR}_{k,n}(\G_m) \cap \cR_{k,T}| \ \bigg| \ A_T^+\right].
	\label{eq:rknkt}
	\end{align}
	Using Lemma \ref{lem:outpath}, we have that
	\begin{align}
	\lim\limits_{m\rightarrow \infty} \E^{\mathrm{o}}  &\left[\frac{1}{m^2}| \overline{\cR}_{k,n}(\G_m) \cap \cR_{k,T}| \ \bigg| \ A_T^+\right] \le \nonumber\\
	&\hspace{1cm} \lim\limits_{m\rightarrow \infty}\E^{\mathrm{o}}  \left[\frac{| \overline{\cR}_{k,n}(\G_m)|}{m^2} \ \bigg| \ A_T^+\right]=0
	\label{eq:Eo:extra}
	\end{align}
	Substituting \eqref{eq:rknkt} in \eqref{eq:Eo:last}, and using \eqref{eq:Eo:extra}, we get
	\begin{align}
	\theta_{k,t}^+(p) &=  \lim_{m \to \infty} \E^{\mathrm{o}}  \left[\frac{1}{m^2}| \cR_{k,T}  \cap \G_m| \ \big| \ A_T^+\right] \nonumber \\
	&= \lim_{m \to \infty}  \E^{\mathrm{o}}  \left[\frac{1}{m^2}| \cR_{k,n}(\Gamma_m) \cap \cR_{k,T} | \ \bigg| \ A_T^+\right] \nonumber \\
	&\stackrel{\text{(a)}}{=} \lim_{m \to \infty}  \E^{\mathrm{o}}  \left[\frac{1}{m^2}| \cR_{k,n}(\Gamma_m) \cap \cR_{k,T} | \ \bigg| \ A_T^+\right] \nonumber\\
	&\hspace{2cm}+ \lim_{m \to \infty} \E^{\mathrm{o}}  \left[\frac{|\cR_{k,n}(\Gamma_m) \cap \cQ|}{m^2} \ \bigg| \ A_T^+\right]  \nonumber\\
	&= \lim_{m \to \infty}  \E^{\mathrm{o}}  \left[\frac{| \cR_{k,n}(\Gamma_m) |}{m^2} \ \bigg| \ A_T^+\right],
	\label{eq:Eo:thet}
	\end{align}
	where the equality labelled (a) above is obtained using \eqref{eq:Q}.
	Upon multiplying \eqref{eq:Eo_sum} by $\frac{1}{m^2}$, and letting $m \to \infty$, we obtain via \eqref{eq:t<k} and \eqref{eq:Eo:thet}:
	\begin{equation*}
	\lim_{m \to \infty} \E^{\mathrm{o}}  \left[\frac{R_{k,n}(\Gamma_m)}{m^2} \right]  = \sum_{t=k}^n \sum_{T \subseteq [n]: \atop |T| = t} \theta_{k,t}^+(p) \, \P^{\mathrm{o}}(A_T^+).
	\end{equation*}
		
	Applying Proposition~\ref{prop:formulap} completes the proof of the first part of the theorem. The second part of the theorem is a consequence of the proposition below.
	\end{IEEEproof}

\begin{prop}
	$$\sum_{t=k}^{n} \sum_{j=k}^{t}\binom{n}{t}\binom{t}{j}(\theta^+(p))^{t+j}(1-\theta^+(p))^{n-j} = \P(Y\ge k),$$
	where $Y\sim \mathrm{Bin}(n,(\theta^+(p))^2)$.
	\label{prop:simpexpr}
\end{prop}
\begin{IEEEproof}
Consider $Y = \sum_{i=1}^n X_iU_i$, where $X_i, U_i$, $i = 1,2\ldots,n$, are i.i.d.\ $\mathrm{Ber}(\theta^+(p))$ random variables. Clearly, each product $X_iU_i$ is $\mathrm{Ber}((\theta^+(p))^2)$, so that $Y\sim \mathrm{Bin}(n,(\theta^+(p))^2)$. 

Alternatively, $\P(Y = j) = \sum_{t=0}^n \P(Y = j \mid X = t) \P(X = t)$, with $X = \sum_{i=1}^n X_i$. Thus,
	\begin{align*}
        \P(Y=j) &= 
	\sum_{t=j}^{n}\binom{t}{j}(\theta^+(p))^j(1-\theta^+(p))^{t-j} \times \\
	& \hspace{2.4cm}\binom{n}{t}(\theta^+(p))^t(1-\theta^+(p))^{n-t}\\
	&= \sum_{t=j}^{n} \binom{n}{t} \binom{t}{j}(\theta^+(p))^{t+j}(1-\theta^+(p))^{n-j}.
	\end{align*}
Hence, 
$$
\P(Y\ge k) = \sum_{j=k}^{n} \sum_{t=j}^{n}\binom{n}{t}\binom{t}{j}(\theta^+(p))^{t+j}(1-\theta^+(p))^{n-j},
$$
from which, upon exchanging the order of the summations, we get the expression in the statement of the proposition.
\end{IEEEproof}

\section{Discussion}\label{sec:discussion}
In this section, we give justifications and heuristics for some of the assumptions made in our analysis.

\subsection{Super-critical region}
\label{sec:supercrit}
Our entire analysis for grids is based on the assumption that we operate in the super-critical region for the site-percolation process. We give an explanation for the same here. Recall that we want values of the forwarding probability $p$ for which the expected fraction of successful receivers,  $\E[\frac{1}{m^2} R_{k,n}(\Gamma_m)]$ is at least $1-\delta$, for some (small) $\delta > 0$. Hence, we need $\E[\frac{1}{m^2} |\cR_{k,n}(\Z^2) \cap \G_m| ]\ge 1-\delta$. If we would like this to hold for all sufficiently large $m$, then $p$ must be such that $\cR_{k,n}(\Z^2)$ has infinite cardinality. This implies, due to the correspondence between probabilistic forwarding and site percolation on $\Z^2$, that $p$ must be such that there exists an infinite (open/extended) cluster in the site percolation process. Thus, we must operate in the super-critical region $p > p_c$. It can also be seen from the simulation results in Figs.~\ref{fig:ergodic_pkndelta} and \ref{fig:ergodic_trans} that $\tau_{k,n,\delta}$ is minimized when $p_{k,n,\delta}$ is in the super-critical region. Further, from Fig. \ref{fig:ergodic_bounds}(a), which provides the minimum forwarding probability obtained numerically from  \eqref{eq:pkndelta_simp}, and which is used to generate the plots in Fig. \ref{fig:ergodic_bounds}(b), it is clear that the expected total number of transmissions is indeed minimized when operating in the super-critical region. We use these arguments as justification for considering only the $p > p_c$ case in our analysis.

\subsection{Insufficiently large $m$} \label{sec:sufflargem}
We now re-visit the disparity seen in Fig.~\ref{fig:ergodic_trans} between the $\tau_{k,n,\delta}$ curves (normalized by the grid size $m^2$) for $\G_{31}$ and $\Gamma_{501}$ obtained via simulations, and the corresponding curve for large grids $\G_m$ obtained via \eqref{eq:taukndelta}. As discussed previously, the numerical evaluation of the RHS of \eqref{eq:taukndelta} relies on the approximation to $p_{k,n,\delta}$ in \eqref{eq:pkndelta_simp}, which, for fixed $k$, $n$ and $\delta$, is valid only for sufficiently large $m$. In the regime where the approximation is not valid (as happens for $n \ge 130$ and $m = 501$ in Fig.~\ref{fig:ergodic_trans}), there is a small discrepancy between the true value of $p_{k,n,\delta}(\G_m)$ obtained via simulations, and the approximation in \eqref{eq:pkndelta_simp}. While this discrepancy is too small to be seen in the plots in Fig.~\ref{fig:ergodic_pkndelta}, it gets blown up when evaluating $\tau_{k,n,\delta}$ using the expression in  \eqref{eq:taukndelta}, which involves $\theta^+(p)$. This blow-up is attributable to the fact that $\theta^+(p)$ exhibits a sharp phase transition around $p=0.6$ (see Fig.~\ref{fig:theta}), so that small changes in $p$ near $0.6$ translate to large changes in $\theta^+(p)$. 

Interestingly, our simulations also indicate that for \emph{any} value of $m$, the true curve for $\frac{1}{m^2} \tau_{k,n,\delta}(\G_m)$ always lies on or above the curve for the ``large-$\G_m$ approximation'' obtained via  \eqref{eq:taukndelta} and \eqref{eq:pkndelta_simp}. We attempt an explanation for this here. We conjecture that the large-$m$ approximation in \eqref{eq:pkndelta_simp} is in fact an inequality valid for all $m$, at least when $\delta$ is small. 

\begin{conj}
Fix $\delta \in (0,1/8)$. Then, for any $k$, $n$ and $m$, we have
\begin{equation}
p_{k,n,\delta}(\Gamma_m) \ \ge \ \inf \{ p \ | \ Pr(Y\ge k) \ge 1-\delta\}, 
\label{ineq:conj}
\end{equation}
where $Y\sim \mathrm{Bin}(n,(\theta^+(p))^2)$. 
\label{conj:pkndelta}
\end{conj}
Thus, assuming the validity of the conjecture, the expected total number of transmissions, $\tau_{k,n,\delta}(\G_m)$, at a forwarding probability equal to $p_{k,n,\delta}(\G_m)$ is at least as large as that when the forwarding probability is set to be equal to the RHS of \eqref{eq:pkndelta_simp} (or \eqref{ineq:conj}). We next provide an argument in support of the conjecture.

\par Recall that 
$$p_{k,n,\delta}(\G_m) = \inf\left\{p \ \bigg| \ \E\left[\frac{R_{k,n}(\G_m)}{m^2}\right]\ge 1-\delta  \right\},$$
while the RHS of \eqref{ineq:conj} is, by virtue of Theorem~\ref{thm:rkn},
$$\inf\left\{p \ \bigg| \ \lim_{m\rightarrow \infty} \E\left[\frac{R_{k,n}(\G_m)}{m^2}\right] \ge 1-\delta  \right\}.$$
Thus, it would suffice to show that when $p$ is large enough to ensure that $\E\left[\frac{R_{k,n}(\G_m)}{m^2}\right]  \ge 1-\delta$, we also have $\lim_{m \to \infty} \E\left[\frac{R_{k,n}(\G_m)}{m^2}\right] \ge \E\left[\frac{R_{k,n}(\G_m)}{m^2}\right]$. This seems to be true: simulation results (see Fig.~\ref{fig:recs_gridsize}) in fact indicate that, for fixed $k$ and $n$, and $p$ sufficiently above criticality, $\E\left[\frac{R_{k,n}(\G_m)}{m^2}\right]$ is an increasing function of $m$.

\begin{figure}
	\centering
	\includegraphics[width=0.5\textwidth]{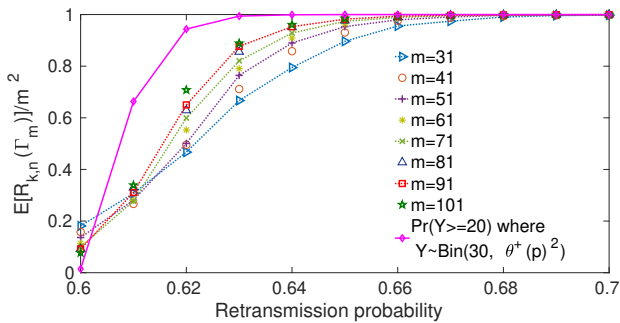}
	\caption{Plot of the expected fraction of nodes that receive at least  $k=20$ out of $n=30$ packets in a $501 \times 501$ grid. Expectation over $100$ iterations.}
	\label{fig:recs_gridsize}
	\vspace*{-1em}
\end{figure}%

\par The intuition behind the increasing nature of the fraction of receivers can be illustrated via the case of $k=1$ and $n=1$. Consider a node $v$ on the boundary of $\G_m$ which receives the sole packet from outside $\G_m$. Let us further suppose that the path through which it receives the packet is contained within $\G_{m+l}$ for some small $l>0$. Node $v$ is not a successful receiver in $\G_m$ but it is successful in $\G_{m+l}$. Additionally, nodes in the $\G_{m+l}$-conduit of $v$ (and the neighbours of these nodes) that are not successful receivers in $\G_m$ become successful receivers in $\G_{m+l}$. Moreover, if node $v$ transmits the packet, there are additional nodes in the interior of $\G_m$ that receive the packet. So, increasing the grid size from $m$ to $m+l$ not only leads to an increase in the number of receivers on the boundary but also results in additional receivers in the bulk. This suggests that the expected number of receivers in $\G_m$ increases in chunks of $m^2$ rather than just $m$. Unfortunately, a rigorous proof of this fact eludes us.

\subsection{ Other graphs}

The analysis on the grid can be extended to other network topologies as well. The ergodic theorems which are detailed in Section \ref{sec:ergthms} of the appendix constitute a key ingredient of our proofs. Similar ergodic theorems are available for other lattice structures as well, like the triangular and hexagonal lattices etc.; 
we refer the reader to \cite{Durrett} and \cite{krengel1985} for further reading on this topic. Our analysis extends to these lattice structures, and we expect finite subgraphs of these lattices to exhibit behaviour similar to that of the grid.

\subsection{ Communication aspects}

For the purpose of analysis, it might be easier to think of each of the packet transmissions happening one after the other in the network. In this scenario, packet collisions are avoided. In a practical implementation, however, it might be that different packets are transmitted on different sub-carriers of an OFDM signal so that interference effects are minimized. Thus, a node could possibly receive different packets from each of its neighbours without any collisions.

\subsection{Algorithm variants}

Several variants of the probabilistic forwarding with coded packets algorithm could be set up and analyzed. For example, the forwarding probability at a node could be a function of its distance from the origin. Alternatively, a node could use more sophisticated means of deciding which received packets it should forward. However, these algorithms require either greater knowledge of the network topology, or they demand additional resources such as buffers or computation capability at the individual nodes. This does not align with our idea of a completely distributed, energy-efficient broadcast algorithm. However, a certain light-weight extension is possible for our model: a node on receiving $k$ out of the $n$ coded packets, can decode the data and subsequently behave as a source for generating additional coded packets which are broadcast. One can reduce the forwarding probability of these secondary sources, and further stipulate that, only those nodes which receive exactly $k$ packets encode and forward packets. Naturally, this is a harder problem to analyze, and we believe that the analysis in this paper will prove to be a stepping stone in understanding such algorithms. 
\par Another minor variant is to ask for $p_{k,n,\delta}$ to be the minimum probability such that the fraction of successful receivers is close to $1$ with a high probability. Simulations using this criterion indicate similar trends for $p_{k,n,\delta}$ and $\tau_{k,n,\delta}$ as in the results presented here.

	
		\section*{Acknowledgements} The research presented in this paper was supported in part by a fellowship from the Centre for Networked Intelligence (a Cisco CSR initiative) of the Indian Institute of Science to the first author, and in part by the DRDO-IISc ``Frontiers'' research programme.
	
	\bibliographystyle{IEEEtran}
	\bibliography{references}

\begin{thebibliography}{10}
\providecommand{\url}[1]{#1}
\csname url@samestyle\endcsname
\providecommand{\newblock}{\relax}
\providecommand{\bibinfo}[2]{#2}
\providecommand{\BIBentrySTDinterwordspacing}{\spaceskip=0pt\relax}
\providecommand{\BIBentryALTinterwordstretchfactor}{4}
\providecommand{\BIBentryALTinterwordspacing}{\spaceskip=\fontdimen2\font plus
\BIBentryALTinterwordstretchfactor\fontdimen3\font minus
  \fontdimen4\font\relax}
\providecommand{\BIBforeignlanguage}[2]{{%
\expandafter\ifx\csname l@#1\endcsname\relax
\typeout{** WARNING: IEEEtran.bst: No hyphenation pattern has been}%
\typeout{** loaded for the language `#1'. Using the pattern for}%
\typeout{** the default language instead.}%
\else
\language=\csname l@#1\endcsname
\fi
#2}}
\providecommand{\BIBdecl}{\relax}
\BIBdecl

\bibitem{tseng2002broadcast}
Y.-C. Tseng, S.-Y. Ni, Y.-S. Chen, and J.-P. Sheu, ``The broadcast storm
  problem in a mobile ad hoc network,'' \emph{Wireless Networks}, vol.~8, no.
  2/3, pp. 153--167, 2002.

\bibitem{sasson2003probabilistic}
Y.~Sasson, D.~Cavin, and A.~Schiper, ``Probabilistic broadcast for flooding in
  wireless mobile ad hoc networks,'' in \emph{Proc.\ IEEE Wireless
  Communications and Networking Conf.\ (WCNC) 2003, vol.\ 2}, March 16--20,
  2003, pp. 1124--1130.

\bibitem{haas2006gossip}
Z.~J. Haas, J.~Y. Halpern, and L.~Li, ``Gossip-based ad hoc routing,''
  \emph{IEEE/ACM Trans.\ Networking}, vol.~14, no.~3, pp. 479--491, 2006.

\bibitem{ncc2018:probfwding}
V.~Kumar B.~R., R.~Antony, and N.~Kashyap, ``The effect of introducing
  redundancy in a probabilistic forwarding protocol,'' in \emph{Proc.\ 2018
  Nat.\ Conf.\ Commun. (NCC 2018), \emph{IIT-Hyderabad}}, Feb 25--28, 2018.

\bibitem{vaze2015random}
R.~Vaze, \emph{Random Wireless Networks}.\hskip 1em plus 0.5em minus
  0.4em\relax Cambridge Univ.\ Press, 2015.

\bibitem{Shen2006dirbroad}
C.-C. Shen, Z.~Huang, and C.~Jaikaeo, ``Directional broadcast for mobile ad hoc
  networks with percolation theory,'' \emph{IEEE Transactions on Mobile
  Computing}, vol.~5, no.~4, pp. 317--332, 2006.

\bibitem{franceschetti2008random}
M.~Franceschetti and R.~Meester, \emph{Random Networks for Communication: From
  Statistical Physics to Information Systems}.\hskip 1em plus 0.5em minus
  0.4em\relax Cambridge Univ.\ Press, 2008.

\bibitem{jogdeo1968monotone}
K.~Jogdeo and S.~Samuels, ``Monotone convergence of binomial probabilities and
  a generalization of {R}amanujan's equation,'' \emph{The Annals of
  Mathematical Statistics}, vol.~39, no.~3, pp. 1191--1195, 1968.

\bibitem{Okamoto1959}
\BIBentryALTinterwordspacing
M.~Okamoto, ``Some inequalities relating to the partial sum of binomial
  probabilities,'' \emph{Annals of the Institute of Statistical Mathematics},
  vol.~10, no.~1, pp. 29--35, Mar 1959. [Online]. Available:
  \url{https://doi.org/10.1007/BF02883985}
\BIBentrySTDinterwordspacing

\bibitem{grimmett}
G.~Grimmett, \emph{Percolation, 2nd ed.}\hskip 1em plus 0.5em minus 0.4em\relax
  Springer-Verlag, 1999.

\bibitem{russo1978note}
L.~Russo, ``A note on percolation,'' \emph{Zeitschrift f{\"u}r
  Wahrscheinlichkeitstheorie und verwandte Gebiete}, vol.~43, no.~1, pp.
  39--48, 1978.

\bibitem{newman1981infinite}
C.~Newman and L.~Schulman, ``Infinite clusters in percolation models,''
  \emph{Journal of Statistical Physics}, vol.~26, no.~3, pp. 613--628, 1981.

\bibitem{Durrett}
R.~Durrett, \emph{Probability: Theory and Examples}.\hskip 1em plus 0.5em minus
  0.4em\relax Cambridge Univ.\ Press, 2013.

\bibitem{krengel1985}
U.~Krengel, \emph{Ergodic Theorems}.\hskip 1em plus 0.5em minus 0.4em\relax de
  Gruyter, 1985.

\bibitem{zubkov2013complete}
A.~M. Zubkov and A.~A. Serov, ``A complete proof of universal inequalities for
  the distribution function of the binomial law,'' \emph{Theory of Probability
  \& Its Applications}, vol.~57, no.~3, pp. 539--544, 2013.

\end{thebibliography}

	
	
	\appendix
	
\subsection{Bounds for the CDF of a Binomial random variable}
\label{apdx:cdfbinom}
The following theorem from \cite{zubkov2013complete} gives tight bounds on the CDF of a binomial random variable in terms of the standard normal CDF. 
\begin{thm}[\cite{zubkov2013complete}, Theorem~1]
	Let $0\le x,p \le 1$ and define $D\left(x\ || \ p\right):=x \ln \frac{x}{p}+(1-x) \ln \frac{1-x}{1-p}, \ \mathrm{sgn}(x):=\frac{x}{|x|}$ for $x\neq 0$, and $\mathrm{sgn}(0):=0$. Let $\{C_{n,p}(k)\}_{k=0}^n$ be defined as follows:
	$$C_{n,p}(0)=(1-p)^n, \ C_{n,p}(n)=1-p^n,$$
	$$C_{n,p}(k)=\Phi \left(\mathrm{sgn}\left(\frac{k}{n}-p\right)\sqrt{2nD\biggl(\frac{k}{n} \left|\right|  p\biggr)}\right), \ 1\le k<n.$$
	For a binomial random variable $X\sim \mathrm{Bin}(n,p)$, for every $k = 0, 1, . . . , n-1$, and for every $p \in (0, 1)$,
	$$C_{n,p}(k)\le \P(X\le k) \le C_{n,p}(k+1).$$
Equalities hold for $k = 0$ and $k = n-1$ only.
	\label{bin_normal_cdf}
\end{thm}

\subsection{Ergodic theorems} \label{sec:ergthms}
	Let $\mathsf{A}$ be a finite alphabet, and $\nu$ a probability measure on it. Consider the probability space $(\Omega,\mathcal{F},\P)$, where $\Omega = \mathsf{A}^{\Z^2}$, $\mathcal{F}$ is the $\sigma$-algebra of cylinder sets, and $\P$ is the product measure $\otimes_{u} \nu_u$ with $\nu_u = \nu$ for all $u \in Z^2$. For $z \in \Z^2$, define the shift operator $T_z: \Omega \to \Omega$ that maps $\omega = {(\omega_u)}_{u \in \Z^2}$ to $T_z\omega$ such that $(T_z\omega)_u = \omega_{u-z}$ for all $u \in \Z^2$. Correspondingly, for a random variable $X$ defined on this probability space, set $T_zX := X \circ T_{-z}$, i.e., $(T_zX)(\omega) = X(T_{-z}\omega)$ for all $\omega \in \Omega$.
	
	The following theorem is a special case of Tempelman's pointwise ergodic theorem (see e.g., \cite[Chapter~6]{krengel1985}). For $\mathsf{A} = \{0,1\}$, this was stated as Proposition~8 in \cite{newman1981infinite}. 
	
\begin{thm}
For any random variable $X$ on $(\Omega,\mathcal{F},\P)$ with finite mean, we have 
$$
\lim_{m \to \infty} \frac{1}{m^2} \sum_{z \in \G_m} T_zX \, = \, \E[X]  \ \ \ \ \ \ \P\text{-a.s.},
$$
where $\G_m := [-\frac{m-1}{2},\frac{m-1}{2}]^2 \cap \Z^2$ is the $m \times m$ grid ($m$ odd).
\label{thm:tempelman}	
\end{thm}
	
	The theorem applies to the case of site percolation, in which $\nu$ above is the Bernoulli($p$) measure on $\mathsf{A} = \{0,1\}$. Applying the theorem with $X = {\mathds{1}}_{\{\0 \in C\}}$, the indicator function of $\0$ being in the (unique when $p > p_c$) IOC $C$, and again with $X = {\mathds{1}}_{\{\0 \in C^+\}}$, we obtain Theorem~\ref{thm:theta}.
	
%
	Next, with $\mathsf{A} = \{0,1\}^n$ and $\nu$ the product of $n$ independent Bernoulli($p$) measures, we are in the setting of $n$ independent site percolations on $\Z^2$. Let $C_{k,n}^+$ be the set of sites that are in the IEC in at least $k$ out of the $n$ percolations. In this case, taking $E$ to be the event that the $\0$ is in the IEC in at least $k$ of the $n$ independent percolations and $X=\mathds{1}_E$, and applying Theorem~\ref{thm:tempelman}, we obtain
	$$\lim_{m \rightarrow \infty}\frac{1}{m^2} |C_{k,n}^+ \cap \G_m| = \P(E) \ \ \ \ \ \ \P\text{-a.s.}$$
	Using the fact that the origin is in the IEC with probability $\theta^+(p)$, and since all the $n$ percolations are independent, the probability on the RHS in the above equation can be evaluated to obtain Theorem~\ref{thm:prodtheta}.

\end{document}